\documentclass[journal,twoside,web]{ieeecolor}
\usepackage{generic}
\usepackage{cite}
\usepackage{amsmath,amssymb,amsfonts}
\usepackage{algorithmic}
\usepackage{graphicx}
\usepackage{algorithm,algorithmic}
\usepackage{hyperref}
\usepackage{textcomp}
\def\BibTeX{{\rm B\kern-.05em{\sc i\kern-.025em b}\kern-.08em
    T\kern-.1667em\lower.7ex\hbox{E}\kern-.125emX}}
\graphicspath{{../tac-fig/}}

\usepackage{booktabs}
\usepackage{siunitx}
\usepackage{nicefrac}
\usepackage{comment}
\usepackage{cite}
\usepackage{siunitx}
\newtheorem{theorem}{Theorem}
\newtheorem{corollary}{Corollary}
\newtheorem{lemma}{Lemma}
\newtheorem{proposition}{Proposition}
\newtheorem{remark}{Remark}
\newtheorem{assumption}{Assumption}

\usepackage[dvipsnames]{xcolor}

\newcommand*\diff{\mathop{}\!\mathrm{d}}
\newcommand{\dxdt}[1]{\dot{x}_{#1}(a,t)}
\newcommand{\dxda}[1]{x'_{#1}(a,t)}
\newcommand{\xat}[1]{x_{#1}(a,t)}
\newcommand{\xstar}[1]{x^*_{#1}(a)}

\newcommand{\psitminusa}[1]{\psi_{#1,t}}

\newlength{\mlLegendThickness}
\newlength{\mlLegendHeight}
\newcommand{\mlLineLegend}[1]{\mbox{\color{#1}
		\protect\rule[\mlLegendHeight]{3mm}{\mlLegendThickness}\hspace*{-1mm}
}}
\newcommand{\mlLineLegendDashed}[1]{\mbox{\color{#1}
		\protect\rule[\mlLegendHeight]{1.5mm}{\mlLegendThickness}\hspace*{0mm}
		\protect\rule[\mlLegendHeight]{1.5mm}{\mlLegendThickness}\hspace*{-1mm}
}}

\definecolor{dunkelblau}{rgb}{0.0, 0.2314, 0.6196}%
\definecolor{hellblau}{rgb}{0.000, 0.7451, 1.0000}%
\definecolor{rot}{rgb}{0.6980,0.1333,0.1333}%

\newcommand{\reviewed}[1]{\textcolor{black}{#1}}

\newcommand{\mycomment}[1]{}

\usepackage{tikz}
\usepackage{lipsum}

\newcommand\copyrighttext{%
  \footnotesize \textcopyright 2012 IEEE. Personal use of this material is permitted.
  Permission from IEEE must be obtained for all other uses, in any current or future 
  media, including reprinting/republishing this material for advertising or promotional 
  purposes, creating new collective works, for resale or redistribution to servers or 
  lists, or reuse of any copyrighted component of this work in other works. 
  DOI: \href{https://doi.org/10.1109/TAC.2025.3589108}{10.1109/TAC.2025.3589108
}}
\newcommand\copyrightnotice{%
\begin{tikzpicture}[remember picture,overlay]
\node[anchor=south,yshift=10pt] at (current page.south) {\fbox{\parbox{\dimexpr\textwidth-\fboxsep-\fboxrule\relax}{\copyrighttext}}};
\end{tikzpicture}%
}

\begin{document}
\copyrightnotice

\title{Stabilization of Predator-Prey Age-Structured Hyperbolic PDE when Harvesting both Species is Inevitable}
\author{Carina Veil, Miroslav Krsti\'c, Iasson Karafyllis, Mamadou Diagne, and Oliver Sawodny%
\thanks{This work was supported by German Research Foundation (Deutsche Forschungsgemeinschaft) under grant SA 847/22-2, project number 327834553. Corresponding author: Mamadou Diagne.}
\thanks{C. Veil is with the Department of Mechanical Engineering, Stanford University, Stanford, CA 94305, USA (e-mail: cveil@stanford.edu).}
\thanks{M. Krsti\'c and M. Diagne are with Department of Mechanical and Aerospace Engineering, University of California San Diego, La Jolla, CA 92093-0411, USA (e-mail: \{krstic, mdiagne\}@ucsd.edu).}
\thanks{I. Karafyllis is with the Department of Mathematics, National Technical University of Athens, Zografou Campus, 15780 Athens, Greece (e-mail: iasonkar@central.ntua.gr).}
\thanks{O. Sawodny is with the Institute for System Dynamics, University of Stuttgart, 70563 Stuttgart, Germany (e-mail: sawodny@isys.uni-stuttgart.de). }}

\maketitle

\begin{abstract}
Populations (in ecology, epidemics, biotechnology, economics, social processes) do not only interact over time but also age over time. It is therefore common to model them as ``age-structured'' partial differential equations (PDEs), where age is the `space variable.' Since the models also involve integrals over age, both in the birth process and in the interaction among species, they  are in fact integro-partial differential equations (IPDEs) with positive states. To regulate the population densities to desired profiles, harvesting is used as input. But non-discriminating harvesting, where wanting to repress one (overpopulated) species will inevitably repress the other (near-extinct) species as well, the positivity restriction on the input (no insertion of population, only removal), and the multiplicative (nonlinear) nature of harvesting, makes control challenging even for ordinary differential equation (ODE) versions of such dynamics, let alone for their IPDE versions, on an infinite-dimensional nonnegative state space. 

With this paper, we introduce a design for a benchmark version of such a problem: a two-population predator-prey setup. 
The model is equivalent to two coupled ODEs, actuated by harvesting which must not drop below zero, and strongly (``exponentially'') disturbed by two autonomous but exponentially stable integral delay equations (IDEs). We develop two control designs. With a modified Volterra-like control Lyapunov function, we design a simple feedback which employs possibly negative harvesting for global stabilization of the ODE model, while guaranteeing regional regulation with positive harvesting. With a more sophisticated, restrained controller we achieve regulation for the ODE model globally, with positive harvesting. For the full IPDE model, with the IDE dynamics acting as large disturbances, for both the simple and saturated feedback laws we provide explicit estimates of the regions of attraction. Simulations illustrate the nonlinear infinite-dimensional solutions under the two feedbacks. The paper charts a new pathway for control designs for infinite-dimensional multi-species dynamics and for nonlinear positive systems with positive controls. 
\end{abstract}

% \begin{IEEEkeywords}
% Enter key words or phrases in alphabetical order, separated by commas. Using the IEEE Thesaurus can help you find the best standardized keywords to fit your article. Use the thesaurus access request form for free access to the IEEE Thesaurus: \underline{https://www.ieee.org/publications/services/thesaurus-acce}\\
% \underline{ss-page.com.}
% \end{IEEEkeywords}

\section{Introduction}
\begin{figure}
    \centering
   \includegraphics[width=0.9\columnwidth]{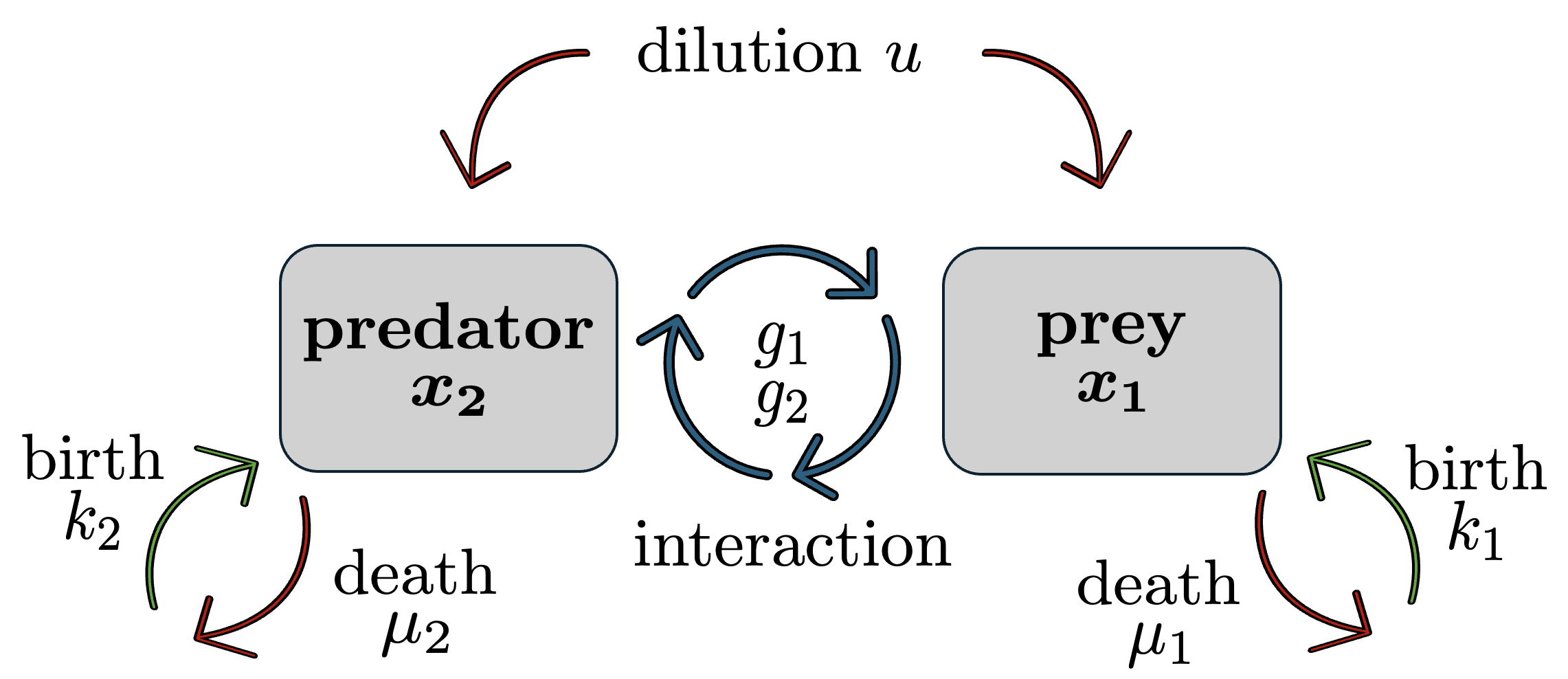}
    \caption{The prey $x_1$ and predator $x_2$ interact via the terms $g_i$. Each species is affected by mortality $\mu_i$ and new population can only enter the system through birth $k_i$. The dilution input $u$ has a repressive effect on both species: harvesting both species is inevitable and represents a challenge for stabilization.}
    \label{fig:ppmodel}
\end{figure}

\IEEEPARstart{T}{o} explain fluctuations of past living populations or predict their future growth, age-structured models serve diverse and rich scientific purposes to describe the evolution of biologically renewable populations over time. 
Structured around age cohorts and involving integrals over age, these models lead to a set of nonlinear integro-partial differential equations (IPDEs) with positive states, describing the dynamics of sub-populations coupled by the law of mass action, a law that governs their contact rate or mixing capabilities \cite{heesterbeek2005law}. Further, the population dynamics of each species is affected by the natural birth and death rates, which may vary depending on age and, potentially, time. 

Many processes in biotechnology, demography or biology exhibit behavior that can be modeled as age-structured population dynamics \cite{inaba2017age, martcheva2015introduction, brauer2012mathematical, rong2007mathematical, gyllenberg1983stability}. 
For instance, the link between demography and economics, the prediction of  the impact of demographic shifts on workforce dynamics, healthcare demand, and social security systems can be achieved by exploiting age-structured population models parameterized by an age- and time-dependent consumption and total value in assets \cite{freiberger2024optimization,tahvonen2009economics}. These models are also exploited to understand the contagion of criminal behaviors in carceral environments and allow to implement age-structured correctional intervention measures \cite{ibrahim2022mathematical, sooknanan2023criminals}, or in other cases, favor the development of more equitable educational strategies when facing a population growth or demographic shifts \cite{lutz2005toward}.

%To improve the performance of bioreactors, t
The behavior of such populations can be studied through chemostat models, where fresh nutrient solution is fed to the biomass-nutrient-mixture at the same rate as it is extracted. To achieve the desired amount of biomass, this rate is used as control input, which makes harvesting and dilution synonyms in this context. 
Whereas single population models are, for example, used to improve the performance of bioreactors for insulin production \cite{smith2013chemostats, baeshen2014cell}, more advanced models describing multiple interacting populations are of special interest for wastewater treatment \cite{canale1969predator, andrews1974dynamic} or protein synthesis, where plasmid-bearing and plasmid-free bacteria compete \cite{li2008competition, xiang2007model}. So far, said processes are mainly modelled with ordinary differential equations (ODEs). Age-structured or spatially structured population models with more than one population are encountered in ecological models or epidemics \cite{holmes1994partial, nhu2018stochastic, nguyen2020analysis}, but their control related studies are limited.  
The main challenge lies in the input to the system being the inevitable harvesting of all species, both overpopulated and underpopulated ones, which poses a challenge for stabilization.

\subsection{Related Work: Control of Population Models}
Generally consisting of a single population, chemostats are very interesting from a control perspective, as they present as nonlinear control problems with input constraints (positive dilution) and inequality state constraints (positive populations) \cite{de2001stabilization}. 
Actuation by dilution in chemostat interacting populations is analogous to the use of pesticides or insecticides in an ecological predator-prey situation \cite{san1974optimal}. Here, for clarity, we use the term ``dilution" to align with terminology from our previous works developed along similar lines. 

Existing studies on nonlinear infinite-dimensional population models stratified by age cohorts often focus on Susceptible-Infected-Recovered (SIR)-like epidemic models and are primarily limited to open-loop stability analysis. These studies typically guarantee asymptotic convergence under restrictive conditions on model parameters \cite{inaba1990threshold,liu2015global}. In addition to the optimal control method  proposed in \cite{albi2021control,boucekkine2011optimal, feichtinger2003optimality}, a relatively old contribution developed a pharmaceutical interventions feedback law, associated to the rate of vaccination to control ``reduced-order" age-structure SIR model for Varicella and Herpes Zoste \cite{Allen}. In recent years, the link of chemostat models to epidemics became particularly prevalent to study the spread of infectious diseases, i.~e., considering the biomass as the infected population and the dilution as the treatment for the disease \cite{smith2013chemostats}. 

The nonlinear infinite-dimensional change of variables introduced in \cite{karafyllis2017stability} has been the basis of most control design and stability analysis of \textbf{single population models}, including the effects of intraspecific competition or actuator dynamics \cite{kurth2023control, haacker2024stabilization}. 
Essentially, the output feedback law contribution from \cite{karafyllis2017stability} stabilizes a single age-structured population system that is neutrally stable in open-loop by introducing self-competition terms in the closed-loop system.
Single population models are, for example, used to control the dilution rate of the chemostat such that the biomass follows certain trajectories in order to maximize the yield of the process \cite{schmidt2018yield, kurth2021tracking, kurth2021optimal, kurth2023control}.

\textbf{Multi-species population models} with age structure, leading to coupled IPDEs, have received little attention in the literature. 
In previous works, we presented a model of two interacting populations in a chemostat with both intra- and interspecific competition terms \cite{kurth2022model}. Key similarities between this work and \cite{kurth2022model} lie in the definition of the contact rate, which is derived from the law of mass action \cite{heesterbeek2005law}. In both cases, the distributed ``reaction term" influences both sub-populations, with mutual interactions between them. Additionally, both models require the inevitable harvesting of both populations as input, which makes stabilization challenging as diluting the entire population may promote the stability of the overpopulated species while possibly having unintendend, detrimental effects on the stability of the underpopulated species.
However, there remains important key distinctions between the two problems:
The IPDE model in \cite{kurth2022model} includes additive self-competition terms that enhance the stability of age-dependent PDE population model and may result open-loop stability of the family of equilibria.
Such stabilizing internal feedback is absent in the present work.
Instead of competition terms, our marginally stable model has skew-symmetric interactions of the two interacting populations in a predator-prey setup, which leads to steady-state oscillations with imaginary open-loop poles.
Last but not least, \cite{kurth2022model} deals with a tracking problem while the present contribution focuses on a setpoint stabilization problem.

\subsection{Contributions}

The first contribution we list here is an obvious one: this is the first set of results on stabilization of age-structured population dynamics with more than one species. Article \cite{kurth2022model} deals with open-loop motion planning for two species but does not pursue stabilization.

% Difficulty of the problem and a quick word on the model
To develop the first stabilizing feedback for multi-species age-structured population dynamics, we face several challenges in the control design and analysis. The system consists of IPDEs, which are nonlinear, must maintain positive states, and the control applied to them must remain positive. Additionally, the dilution actuation, which is not species-specific, is a challenge for control design, as one must either harvest both species or neither, regardless of the imbalance among them. But the greatest of innovation required in the control design, relative to the existing feedback design methods, is that the control must remain positive, namely saturated from below. The limits to global stabilizability of systems with saturated control are well known, including that the open-loop plant must not be exponentially unstable. Our predator-prey model is not exponentially unstable indeed, but it is nothing like feedforward systems for which saturation-based designs exist. The model doesn't even involve any integrators chains---the basis for stabilizability of both feedforward and strict-feedback systems. We introduce both new Control Lyapunov functions (CLFs) and new nonlinear positive feedback laws for these positive infinite-dimensional nonlinear systems.

Our design methodology uses state transformations, as in \cite{karafyllis2017stability}, to convert the nonlinear predator-prey IPDE system into two coupled nonlinear ordinary differential equations (ODEs). 
Unlike the single-species model in \cite{karafyllis2017stability}, the ODE for two interacting species is strongly (``exponentially'') disturbed by two autonomous but exponentially stable integral delay equations (IDEs). 

For such an ODE+IDE system, we present two control designs, with two theorems established for each of the designs. The first control design is a simple $L_gV$ type of a design for the ODE model and for a CLF we construct. Due to its simplicity, the first controller achieves global stability but without a guarantee that its input remains positive for all initial conditions. For the IPDE model to which this first controller is applied, we estimate the region of attraction in this system's infinite-dimensional state space. 

Our second feedback design is more complex, to guarantees positivity of the input. For the ODE model, the second controller is globally stabilizing, whereas for the IPDE model we estimate the origin's region of attraction. 

With these two control designs, and the resulting four theorems, we present a repertoire of possibilities in feedback synthesis and stability analysis for age-structured multi-population models. What we don't achieve, as it is impossible because of the disturbing effect of the IDE dynamics on the ODE model, is global stability for the IPDE model. Region of attraction estimates is the best one can get for stabilization of multi-species age-structured population dynamics.

For the predator-prey system with a dilution input, a small innovation in the Lypunov function development suffices, relative to the existing
% How do we differ from existing Literature? (should this be placed here or in 
Volterra Lyapunov (VL) functions. The VL functions employ exponential and linear functions of the state variables for the study of stability of Lotka-Volterra population models \cite{hofbauer1998evolutionary}. Lyapunov functions of that form also appear in \cite{malisoff2016stabilization}, which is inspired by populations models but not dealing directly with models recognized in the literature as actual population dynamics. Our CLFs in this paper have an appearance of VL functions, but they include crucial modifications, in the form of weighing more strongly the predator state relative to the prey state. The CLFs we design are unusable for the study of (the neutral) stability of the open-loop predator-prey system but they are crucial for enabling our design of an asymptotically stabilizing feedback for the system.

\paragraph*{Added Value Relative to the Paper's Conference Version}
This article is a journal version of our conference submission \cite{veil2025stabilization}, which  contains a small subset of the  results given here.
The journal version's \emph{added value} includes about 9 pages of additional material: more advanced feedback design with positive dilution control (Theorems \ref{thm:global-stabilization-controller} and \ref{thm:Iasson}), as negative dilution amounts to introducing population to the system, is not practically feasible, more detailed proofs,  discussions, and, lastly, more simulations.

\paragraph*{Organization}
This paper is organized as follows: 
Section~\ref{sec:system-model} introduces the model, its steady-states and the system transformation. 
In Section~\ref{sec:control-design}, we consider the ODE system, provide an open-loop stability analysis and present a first CLF based feedback law with unrestricted dilution that globally stabilizes the ODE system, but only ensures positive dilution regionally. In Section~\ref{sec:stability-nonzero-psi}, we establish that this CLF-based control law also locally stabilizes the full ODE-IDE system and provide an explicit estimate of the region of attraction. 
To guarantee positive dilution, we introduce an enhancement of our control law in Section~\ref{sec:stability-positive-dilution}, and show in Section~\ref{sec:regional-stabilization-positive-dilution} that it also locally stabilizes the full ODE-IDE system while the dilution input remains positive at all times. Finally, simulations are presented in Section~\ref{sec:simulations}. Section~\ref{sec:conclusion} concludes this article.

\reviewed{\paragraph*{Notation}
\(\mathcal{K}\) is the class of all strictly increasing functions \( a \in C^0(\mathbb{R}_+;\mathbb{R}_+)\), with \( a(0) = 0 \).
\(\mathcal{KL}\) is the class of functions \( \beta : \mathbb{R}_+ \times \mathbb{R}_+ \to \mathbb{R}_+ \) which satisfy the following: For each \( t \geq 0 \), the mapping \( \beta(\cdot,t) \) is of class \( \mathcal{K} \), and, for each \( s \geq 0 \), the mapping \( \beta(s,\cdot) \) is nonincreasing with \(\lim_{t \to \infty} \beta(s,t) = 0 \) (see References \cite{khalil2014, karafyllis2011stability}). Further, the index $i$ will always refer to $i=1,2$ throughout the manuscript.}

\section{System Model}\label{sec:system-model}
The age-structured predator-prey model with initial conditions (IC) and boundary conditions (BC) is given by
\begin{subequations}
	\begin{alignat}{3}
	&\dxdt{1} + \dxda{1} =&&  
	- \xat{1}  \Bigg[ \mu_1(a) + u(t) \nonumber \\
    &&&+ \int_0^A g_1(\alpha) x_2(\alpha,t) \diff \alpha 
    %\langle g_1(\cdot), x_2(\cdot,t) \rangle 
    \Bigg]  \label{eq:ipde1}\\ 
 	&\dxdt{2} + \dxda{2} =&& 
	- \xat{2}  \Bigg[ \mu_2(a) + u(t) \nonumber\\
    &&& + \frac{1}{\int_0^A g_2(\alpha) x_1(\alpha,t) \diff \alpha} 
    %+ \frac{1}{\langle g_2(\cdot), x_1(\cdot,t)\rangle}
    \Bigg] 
     \label{eq:ipde2}\\ 
	&\text{IC}: \qquad\quad \ \ x_i(a, 0) =&&\  x_{i,0}(a), \\
	&\text{BC}: \qquad\quad \, \ x_i(0,t) =&&\ \int_0^A k_i(a) x_i(a,t) \diff a  \label{eq:bc},
 \end{alignat} \label{eq:karafyllis_system}%
\end{subequations}
where, for $i,j \in \{1,2\}$, $i\neq j$, $x_i(a,t)>0$ is the population density, i.~e. the amount of organisms of a certain age $a \in [0,A]$ of the two interacting populations $x_1(a,t)$ and $x_2(a,t)$ with $(a,t) \in \mathbb{R}_+ \times [0,A]$, their derivatives $\dot{x}_i$ with respect to time and $x'_i$ with respect to age, and the constant maximum age $A>0$. 
The interaction kernels $g_i(a):[0,A]\rightarrow\mathbb{R}_0^+$, the mortality rates $\mu_i(a):[0,A]\rightarrow\mathbb{R}_0^+$, the birth rates $k_i(a):[0,A]\rightarrow\mathbb{R}_0^+$ are functions with $\int_0^A \mu_i(a)\diff a > 0$, $\int_0^A g_i(a)\diff a > 0$, $\int_0^A k_i(a)\diff a >0$. The dilution rate $u(t):\mathbb{R}^+\rightarrow \mathbb{R}_0^+$, is an input affecting both species.

System~\eqref{eq:karafyllis_system} describes a predator-prey population dynamics where  $x_1$ is the prey, $x_2$ is the predator. When $x_2$ is large, $x_1$ is being reduced (\ref{eq:ipde1}), and conversely, 
when $x_1$ is small, $x_2$ is being reduced (\ref{eq:ipde2}).
The choice of the response function $(\int_0^A g_2(\alpha) x_1(\alpha,t) \diff \alpha)^{-1}$ is designed for controlled population systems where alternative survival strategies can be neglected, particularly in biotechnological applications.
Importantly, \eqref{eq:karafyllis_system}  %is unstable and 
exhibits \emph{periodic solutions} and its states are functions of $a\in[0,A]$ with values $x_i$ at time $t$, that belong to the function spaces $\mathcal{F}_i$, $i=1,2$, 
\begin{align}
    \mathcal{F}_i = \Big\{ \xi \in PC^1 ([0,A];(0,\infty)):% \nonumber \\ 
    \xi(0) =\int_0^A k_i(a) \xi(a) \diff a
    \Big\}.
\end{align}
For any subset $\mathcal{S}\subseteq \mathbb{R}$ and for any $A>0$, $PC^1([0,A];\mathcal{S})$ denotes the class of all functions $f\in C^0([0,A];\mathcal{S})$ for which there exists a finite (or empty) set $B\subset(0,A)$ such that: (i) the derivative $x'(a)$ exists at every $a \in (0,A)\backslash B$ and is a continuous function on $(0,A) \backslash B$, (ii) all meaningful right and left limits of $x'(a)$ when $a$ tends to a point in $B \cup \{0,A\}$ exist and are finite.

\subsection{Steady-State Analysis}
In order to determine the equilibria of (\ref{eq:karafyllis_system}), the following lemma is needed.
\begin{lemma}[Lotka-Sharpe condition \cite{sharpe1911problem}]
\label{lem:ls}
The equations
\begin{align}\label{eq:lotkasharpe}
    \int_0^A \tilde{k}_i(a) \diff a = 1, \ i=1,2
\end{align}
with
\begin{align}
    \tilde{k}_i(a) =  k_i(a) e^{-\int_0^a \left(\mu_i(s) + \zeta_i \right)\diff s}
\end{align}
have unique real-valued solutions $\zeta_1(k_1,\mu_1)$ and $\zeta_2(k_2,\mu_2)$, which depend on the birth rates $k_1, k_2$ and mortality rates $\mu_1,\mu_2$.
\end{lemma}

\begin{proposition}[Equilibrium]
\label{prop:steady}
The equilibrium  state $(\xstar{1},\xstar{2})$ of the population system (\ref{eq:karafyllis_system}), along with the equilibrium dilution input $u^*$, is given by
\begin{subequations}
\begin{align}
x_i^* (a) &= x_i^*(0)\underbrace{e^{-{\int_0^a (\zeta_i + \mu_i(s)) \diff s}}}_{\tilde x_i^*(a)},\label{eq:ss_profiles}\\%=  x_i^*(0) \xstartildea{i}\\
u^* &= \zeta_1 -  \lambda_2 = \zeta_2 - \frac{1}{\lambda_1}  \in \left(0, \min\left\{\zeta_1 , \zeta_2\right\}\right)\,,
\label{eq:u_star_constraint}
\end{align}
with unique parameters $\zeta_i(k_i,\mu_i)$ resulting from the Lotka-Sharpe condition of Lemma~\ref{lem:ls}, 
\begin{align}
        \lambda_i(u^*) &:=  \int_0^A g_j(a) x_i^*(a) \diff a  > 0 \,,
    \label{eq:def_lambda}
\end{align}
\end{subequations}
and the positive  concentrations of the newborns 
\begin{subequations}
    \begin{align}
    x_1^*(0) &=  \frac{1}{\left(\zeta_2-u^*\right) \int_0^A g_2(a) \tilde{x}^*_1(a) \diff a} >0, \\
    x_2^*(0) &= \frac{\zeta_1-u^*}{\int_0^A g_1(a) \tilde{x}^*_2(a) \diff a}  > 0\,.
    \end{align}
    \label{eq:ic_constraint}
\end{subequations}
\end{proposition}

\begin{remark}\rm
\label{rem:steady}
It is interesting to observe from \eqref{eq:ic_constraint} that, as the equilibrium dilution $u^*$ grows, the equilibrium population of the predators decreases, unsurprisingly, but the equilibrium population of the prey increases. The greater the collective harvesting, the more the prey benefit! 
\hfill$\Box$
\end{remark}

\begin{proof}[Proof of Proposition \ref{prop:steady}.]
Neglecting the time dependence in (\ref{eq:karafyllis_system}) results in the constant dilution rate $u^*$ and interaction terms, which are merged in the parameters $\zeta_i$
\begin{subequations}\label{eq:ss}
\begin{align}
 0 &= -x_i^{*'} (a) - \xstar{i} \left(\mu_i(a) +  \zeta_i \right) =: \mathcal{D}_i^* x_i^*(a) \label{eq:ss_ode} \\
 \zeta_1 &:= u^*  + \int_0^A g_1(a) x_2^*(a) \diff a  = u^* + \lambda_2\\
 \zeta_2 &:= u^* + \frac{1}{\int_0^A g_2(a) x_1^*(a) \diff a} = u^* + \frac{1}{\lambda_1}\\
 x_i^*(0) &= \int_0^A k_i(a) x_i^*(a) \diff a.
\end{align}%
\end{subequations} %
The steady-state profiles (\ref{eq:ss_profiles}) result from solving (\ref{eq:ss_ode}) defined with the differential operator $\mathcal{D}_i^*$ for arbitrary initial conditions $ x_i^*(0)$. Inserting the solutions into the boundary conditions (\ref{eq:bc}) results in the Lotka-Sharpe condition and unique real-valued parameters $\zeta_i$. The definition of $\zeta_i$ contains the same steady-state input $u^*$ for $i \in \{1,2\}$. Equating both conditions (\ref{eq:u_star_constraint}), introducing the parameters $\lambda_i$ (\ref{eq:def_lambda}),
and solving for the initial conditions restricts possible steady-state by constrained initial conditions $x_i^*(0)$ (\ref{eq:ic_constraint}) that ensure a positive steady-state dilution $u^*$.
\end{proof}

\begin{figure*}[t!p]
    \centering
    \begin{minipage}{0.49\textwidth}
    \includegraphics[width=\columnwidth]{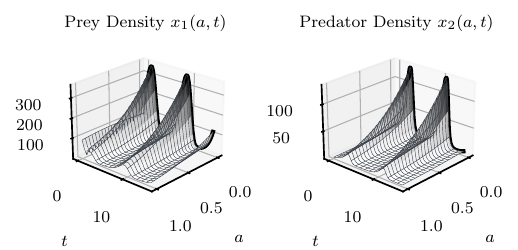}
    \end{minipage}
    \begin{minipage}{0.49\textwidth}
    \includegraphics[width=\columnwidth]{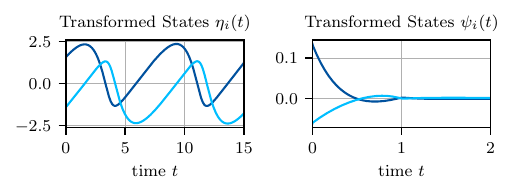}
    \end{minipage}
    \caption{Periodic behavior of the population densities $x_i$ when $u=u^*$, along with the states of the controllable ODE system $\eta_1$ (\mlLineLegend{dunkelblau}), $\eta_2$ (\mlLineLegend{hellblau}) and the autonomous but stable IDEs $\psi_1$ (\mlLineLegend{dunkelblau}), $\psi_2$ (\mlLineLegend{hellblau}). The parameter set (\ref{eq:parameters}) and ICs (\ref{eq:ic-1}) used are specified in Section~\ref{sec:simulations}.}
    \label{fig:simulation-uncontrolled-ic1}
\end{figure*}

\subsection{System Transformation}
In a next step, the PDE system (\ref{eq:karafyllis_system}) is transformed into two coupled ODEs that are actuated by the dilution rate, and two autonomous but exponentially stable IDEs as a basis for designing a stabilizing feedback law.
To get there, the relationship between hyperbolic PDEs and integral delay equations (IDEs) is exploited.

\begin{lemma}\label{lm:solutions}
    Every solution of the population system (\ref{eq:karafyllis_system}) is of the form
    \begin{align}
        x_i(a,t)&=e^{-\int_0^a\mu_i(\alpha)\diff \alpha+\int_{t-a}^t w_{1,i}(\tau)\mathrm{d}\tau}w_{2,i}(t-a), \label{eq:solution-x-ides}
    \end{align}
    and corresponds to a solution of the coupled IDEs
    \begin{subequations}
    \begin{align}
        &w_{1,1} = -u(t)
         \nonumber \\
        &-\int_0^A g_1(\alpha) 
        e^{-\int_0^a \mu_2(s) \diff s 
        + \int_{t-\reviewed{a}}^{\reviewed{a}} w_{1,2}(\reviewed{s})\diff \reviewed{s}}w_{2,2}(t-\reviewed{a}) \diff \reviewed{a}, \\
        &w_{1,2} = -u(t) \nonumber \\
        &- \frac{1}{\int_0^A g_2(\reviewed{a}) e^{-\int_0^a \mu_1(s) \diff s + \int_{t-\reviewed{a}}^{\reviewed{a}} w_{1,1}(\reviewed{s})\diff \reviewed{s}} w_{2,1}(t-\reviewed{a}) \diff \reviewed{a}}, \\
        w_{2,i} &= \int_0^A k_i(\reviewed{a}) e^{-\mu_i(s) \diff s + \int_{t-\reviewed{a}}^{\reviewed{a}} w_{1,i}(\reviewed{s}) \diff \reviewed{s}} w_{2,i}(t-\reviewed{a}) \diff \reviewed{a} \label{eq:ide-v2i}
    \end{align}
    \end{subequations}
    with initial conditions $w_{1,i}(-q)=0$ and $w_{2,i}(q)=e^{\int_0^q \mu_i(s) \diff s }x_{i,0}(q)$ for $q \in [0,A]$.
\end{lemma}

\begin{proof}
Redefining the input as 
\begin{subequations}
   \begin{align}
    \tilde{u}_1(t) &= u(t) + \int_0^A g_1(a) x_2(a,t) \diff a,\\
    \tilde{u}_2(t) &= u(t) + \frac{1}{\int_0^A g_2(a)x_1(a,t) \diff a}
    \end{align}  
\end{subequations}
decouples the dynamics. Then, Theorem 3.1 from  \cite{karafyllis2014relation} is applied, defining the IDEs in dependence of hyperbolic PDEs with arbitrary parameters. In a last step, the original input is resubstituted.
\end{proof}

Now, this relationship between hyperbolic PDEs and IDEs from Lemma~\ref{lm:solutions} is exploited to split off a two-dimensional ODE from the infinite-dimensional IPDEs (\ref{eq:ipde1}) and (\ref{eq:ipde2}) by a system transformation.
\begin{proposition}[System Transformation]\label{prop:trafo}
    Consider the mapping
    \begin{align}
    \begin{bmatrix}
    \eta_1(t)\\\eta_2(t)\\\psi_1(t-a)\\\psi_2(t-a)
    \end{bmatrix}=
    \begin{bmatrix}
        \mathrm{ln}\left(\Pi_1[x_{1}](t)\right)\\
        \mathrm{ln}\left(\Pi_2[x_{2}](t)\right)\\
        \frac{x_{1}(a,t)}{x^*_1(a)\Pi_1[x_{1}](t)}-1\\
        \frac{x_{2}(a,t)}{x^*_2(a)\Pi_2[x_{2}](t)}-1
    \end{bmatrix}. \label{eq:system-trafo}
\end{align}
defined with the functionals 
\begin{align}
    \Pi_i[x_i](t)
    &= \frac{\int_0^A \pi_{0,i}(a) x_i(a,t)  \diff a}
    {\int_0^A a k_i(a) x_i^*(a) \diff a}, \label{eq:pi-function}
\end{align}
where
\begin{align}
    \pi_{0,i}(a) =\int_a^A k_i(s) e^{{\int_s^a \zeta_i + \mu_i(l) \diff l}} \diff s
\end{align}
are the adjoint eigenfunctions to the zero eigenvalue of the adjoint differential operator \cite{schmidt2018yield} 
\begin{align}
    \mathcal{D}_i^* \pi_{0,i}(a) = \frac{d\pi_{0,i}(a)}{da} - (\mu_i(a) + \zeta_i) \pi_{0,i}(a) + k_i(a) \pi_{0,i}(0).
\end{align}
The transformed variables satisfy the transformed system
\begin{subequations}
\label{eq:transformed-system}
    \begin{align}
        \dot{\eta}_1(t) =& \zeta_1 - u(t) \nonumber \\
        &- e^{\eta_2(t)} \int_0^A g_1(a) \reviewed{x_2^*(a)} \left(1 + \psi_2(t-a)\right) \diff a \label{eq:doteta1} \\
        %- e^{\eta_2(t)} \langle g_1(\cdot), x_2^*(\cdot) \left(1 + \psi_2(t-a)\right) \rangle \label{eq:doteta1} \\
        \dot{\eta}_2(t) =& \zeta_2 - u(t) \nonumber \\
        &- \frac{e^{-\eta_1(t)}}{\int_0^A g_2(a) \reviewed{x_1^*(a)} \left(1 + \psi_1(t-a)\right) \diff a} \label{eq:doteta2} \\
        \psi_i(t) 
        =& \int_0^A \tilde{k}_i(a) \psi_i(t-a) \diff a \label{eq:psi-def}\\
        \eta_i(0) =& \ln \left( \Pi [x_{i,0
        }]\right)=:\eta_{i,0} \label{eq:eta0}\\
        \psi_i(-a) =& \frac{x_{i,0}(a)}{x_i^*(a)\Pi [x_{i,0
        }]} - 1=: \psi_{i,0}(a)\label{eq:psi0}
    \end{align}
\end{subequations}
\reviewed{with $i=1,2$,} and, the unique solution of system (\ref{eq:karafyllis_system}) is given by
\begin{align}
    x_i(a,t) = x_i^*(a)e^{\eta_i(t)}(1+\psi_i(t-a)).\label{eq:solution-x}
\end{align}
\end{proposition}

\reviewed{Let's now explain the importance of transformation~\eqref{eq:system-trafo} and its significance for the existence/uniqueness of solutions problem for~\eqref{eq:karafyllis_system}. Proposition~\ref{prop:trafo} shows that the existence/uniqueness of solutions problem for~\eqref{eq:karafyllis_system} is equivalent to the existence/uniqueness of solutions problem for~\eqref{eq:transformed-system}. Existence, uniqueness, and basic estimates for the $\psi$-components of the solution, i.e., existence and uniqueness for~\eqref{eq:psi-def}, \eqref{eq:psi0}, is guaranteed by the results contained in Section 4 of the paper \cite{karafyllis2017stability} (which exploits existence/uniqueness results for linear neutral delay equations in \cite{hale2013introduction}). Having obtained in a unique way the $\psi$-components of the solutions for all $t \geq 0$, we can deal with the existence/uniqueness problem for~\eqref{eq:doteta1}, \eqref{eq:doteta2}, and \eqref{eq:eta0}. The ODE system~\eqref{eq:doteta1}, \eqref{eq:doteta2} can be seen as a standard time-varying ODE system (where time variance is a result of the substitution of the $\psi$-components of the solutions). Therefore, we can say that every input $u$ (closed-loop or open-loop) that guarantees existence/uniqueness of solutions for~\eqref{eq:doteta1}, \eqref{eq:doteta2} and $t \geq 0$ also guarantees existence/uniqueness of solutions for~\eqref{eq:karafyllis_system}.}

\begin{proof}[Proof of Proposition~\ref{prop:trafo}.]
    Deriving the states $\eta_i$ with respect to time yields
    \begin{align}
        \dot{\eta}_i(t) = \frac{\dot{\Pi}_i[x_i](t)}{\Pi_i[x_i](t)} = \frac{\int_0^A \pi_{0,i}(a) \dot{x}_i(a, t) \diff a}{\int_0^A \pi_{0,i}(a) x_i(a, t) \diff a}, \label{eq:etadot-eigenfunction}
    \end{align}
 whose numerators can be rewritten as 
 \begin{align}
     \int_0^A \pi_{0,i}(a) \dot{x}_i(a, t) \diff a = \int_0^A \pi_{0,i}(a), x_i(a, t) \diff a \left(w_{1,i}(t) + \zeta_i \right)
 \end{align}
    by applying Green's Lemma. Inserting this result in (\ref{eq:etadot-eigenfunction}) yields
    \begin{align}
        \dot{\eta}_i(t) = \zeta_i + w_{1,i}(t) 
    \end{align}
    and allows for the determination of the states
    \begin{align}
        \eta_i(t) = \eta_{i,0} + \int_0^t \left(\zeta_i + w_{1,i}(\tau) \right)\diff \tau
    \end{align}
    by integration with respect to time and the IC $\eta_{i,0}$. Hence, the IDE (\ref{eq:ide-v2i}) can be expressed as 
    \begin{align}
        w_{2,i}(t) = \int_0^A \tilde{k}_i(a)e^{\reviewed{\int_{t-a}^t }\dot{\eta}_i(\reviewed{s}) \diff \reviewed{s}} w_{2,i}(t-a) \diff a,
    \end{align}
    resulting in 
    \begin{align}
        e^{-\eta_i(t)} w_{2,i}(t) = \int_0^A \tilde{k}_i(a) e^{-\eta_i(t-a)} w_{2,i}(t-a) \diff a. \label{eq:temp-eq-proof}
    \end{align}
    Note that, by definition from (\ref{eq:system-trafo}), the states of the internal dynamics are given by
    \begin{align}
        \psi_i(t) = \frac{e^{-\eta_i(t)} w_{2,i}(t)}{x_i^*(0)} -1. \label{eq:temp-eq-proof2}
    \end{align}
    Plugging in the results from (\ref{eq:temp-eq-proof}) into (\ref{eq:temp-eq-proof2}) yields the proposed IDE (\ref{eq:psi-def}). In a final step, inserting (\ref{eq:temp-eq-proof}) resorted after $w_{2,i}(t)$ in the solutions of the population systems (\ref{eq:solution-x-ides}) from Lemma~\ref{lm:solutions} results in (\ref{eq:solution-x}).
    \end{proof}

A simulation illustrating the periodic behavior of the uncontrolled system in the sense $u=u^*$ is shown in Figure~\ref{fig:simulation-uncontrolled-ic1} in $x$- and $\eta$-, $\psi$-states.

Before we deal with the control design for the transformed system, some helpful properties of the IDEs are stated. The interested reader is referred to \cite{karafyllis2017stability} for further derivations. 
Note that it was proved in \cite{karafyllis2017stability} that the states $\psi_i$ of the internal dynamics are restricted to the sets
\begin{align}
    \mathcal{S}_i= \Big\{ &\psi_i \in C^0([-A,0];(-1,\infty)): \nonumber \\
    &P(\psi_i)=0 \land \psi_i(0) = \int_0^A \tilde{k}_i(a) \psi_i(-a) \diff a \Big\}
\end{align}
where
\begin{align}
    P(\psi_i) = \frac{\int_0^A \psi_i(-a) \int_{a}^A \tilde{k}_i(s) \diff s \diff a}{ \int_0^A a \tilde{k}_i(a) \diff a },
\end{align}
and that the states $\psi_i$ of the internal dynamics are globally exponentially stable in the $\mathcal{L}^{\infty}$ norm, which means that there exist $M_i\geq1$, $\sigma_i\geq0$ such that
\begin{align}
\label{eq:psi-decays}
    \psi_i(t-a) \leq M_i e^{-\sigma_i t}||\psi_{i,0}(a)||_{\infty} 
\end{align}
holds for all $t\geq 0$ and $\psi_{i,0}\in \mathcal{S}_i$ and all $a\in[-A,0]$.

Please note that in the following we drop the argument $(t)$ in the state $\eta_i(t)$, $\psi_i(t)$ and use the notation $\psi_{i,t}:=\psi_i(t-a)$ to denote the "age-history" of $\psi_i$ at certain $t\geq0$ for improved readability.
Furthermore, we introduce the vector states
$\eta:=[\eta_1, \eta_2]$, 
$\psi:=[\psi_1, \psi_2]$, 
%$\psi_t := [\psi_{1,t}, \psi_{2,t}]^\top$ 
for a more concise notation.

\section{Control Design with Unrestricted Dilution}\label{sec:control-design}
As the internal dynamics are stable, instead of (\ref{eq:doteta1}), (\ref{eq:doteta2}), for the sake of control design, we first consider the ODE system with $\psi_i \equiv 0$,
\begin{subequations}\label{eq:ode-system-before-phi}
    \begin{align}
        \dot{\eta}_1 &= \zeta_1 - u - \lambda_2 e^{\eta_2},    \\
        \dot{\eta}_2 &= \zeta_2 - u - \frac{1}{\lambda_1} e^{-\eta_1}.
    \end{align}
\end{subequations}
Using the constraint (\ref{eq:u_star_constraint}) linking the steady-state dilution $u^*$ and the parameters $\zeta_i$, namely, $ u^* =\zeta_1 - \lambda_2 = \zeta_2 - \frac{1}{\lambda_1}$, with the functions
\begin{subequations}
\label{phi-def}
\begin{align}
    \phi_1(\eta_1) &:= \frac{1}{\lambda_1} (1- e^{-\eta_1}) \leq \frac{1}{\lambda_1},\\
%\phi_2(\eta_2)&:=\lambda_2 (1-e^{\eta_2})  \leq \lambda_2 \\
%    \textcolor{red}{
\phi_2(\eta_2)&:=\lambda_2(e^{\eta_2}-1)
%} 
\geq -\lambda_2
\end{align}
\end{subequations}
we write (\ref{eq:ode-system-before-phi}) as
\begin{subequations}
\label{eq:ode-system}
    \begin{align}
        \dot{\eta}_1 &= u^* - u - \phi_2(\eta_2) \\
        \dot{\eta}_2 &= u^* - u + \phi_1(\eta_1).
    \end{align}
\end{subequations}
Since $\phi_1$ is increasing in the prey concentration $\eta_1$, the presence of prey clearly enhances the predator population. In contrast, since $\phi_2$ is decreasing in the predator concentration $\eta_2$, the predator has a repressive effect on the prey population. Additionally, dilution $u$ also has a repressive effect---on both populations. Hence, the dilution and predator work in tandem relative to the prey population, posing a risk of overharvesting. As we shall see, this compels a choice of a feedback law that may take negative dilution values, to compensate for the fact that, when the prey population is depleted, positive dilution may result in overharvesting the prey and, consequently, in the extinction of both populations. 

\subsection{Feedback Linearizing Design}

The system \eqref{eq:ode-system-before-phi} is feedback linearizable and, before we produce our preferred design, inspired by the system's structure, we briefly explore the feedback linearizing option. Additionally, there are multiple backstepping design options, which we leave out of consideration. Taking the change of variables 
\begin{eqnarray}
    y&=&\eta_1-\eta_2\\
    %z&=&\phi_2(\eta_2) - \phi_1(\eta_1)\,,
    z&=&- \phi_1(\eta_1) - \phi_2(\eta_2)\,,
\end{eqnarray}
which is a global diffeomorphism, and feedback
\begin{align}
\label{eq-fbk-lin}
    u(t) = u^* &+ \Big(\lambda_2 e^{\eta_2} + \frac{1}{\lambda_1} e^{-\eta_1}\Big)^{-1} \Big( -k_1 \left(\eta_1-\eta_2\right) \nonumber \\
    &+ k_2 
    \left(\phi_1(\eta_1) + \phi_2(\eta_2)\right)
    + \lambda_2 e^{\eta_2}\phi_1(\eta_1) \nonumber \\
    &- \frac{1}{\lambda_1} e^{-\eta_1}\phi_2(\eta_2) \Big) 
\end{align}
with $k_1,k_2>0$, we arrive at the closed loop system 
\begin{subequations}
   \begin{align}
        \dot{y} =& z \\
        \dot{z}=& -k_1y - k_2z.
    \end{align} 
\end{subequations}
We abandon this approach because \eqref{eq-fbk-lin} is not only complicated but the set in the plane $(\eta_1,\eta_2)$ in which the dilution $u$ remains positive is very complicated.

\subsection{Open-Loop Stability}

Before we proceed to feedback stabilization, we perform open-loop stability analysis. We introduce the functions
\begin{subequations}
\label{Phi-defs}
\begin{align}
    \Phi_1(\eta_1) := \int_0^{\eta_1} \phi_1(\beta) \mathrm{d}\beta 
    &= \frac{1}{\lambda_1}( e^{-\eta_1} - 1+\eta_1) \nonumber \\
    &= -\phi_1(\eta_1)+ \frac{1}{\lambda_1}\eta_1.\\
    \Phi_2(\eta_2) := \int_0^{\eta_2} \phi_2(\beta) \mathrm{d}\beta 
    &= \lambda_2(e^{\eta_2} - 1-\eta_2) \nonumber \\
    &= \phi_2(\eta_2) - \lambda_2\eta_2
\end{align}
\end{subequations}
and note that  $\Phi_1(0)=\Phi_2(0)=0$, as well as that, for $r\neq 0$,
$\Phi_1(r)>0$  and $\Phi_2(r)>0$, and $\lim_{r\to\pm\infty}\Phi_1(r)\to\infty$, $\lim_{r\to\pm\infty}\Phi_2(r)\to\infty$. 
We use these functions as Lyapunov candidates but, before doing so, let us examine $(\Phi_1,\Phi_2)$ in some detail. We note that $\Phi_2$ gives greater weight to the predator surplus (exponential) than to predator deficit (linear), whereas the function $\Phi_1$ gives a greater weight to the prey deficit (exponential) than to prey surplus (linear). Hence both functions grow only linearly (and not exponentially) when predator deficit and prey surplus are exhibited, namely, when the prey is highly advantaged over the predator, namely, in the third quadrant of the $(\eta_1,\eta_2)$-plane. This has consequences on performance, making convergence under feedback harder to achieve when the prey is initially advantaged (fourth quadrant) than when the predator is initially advantaged (second quadrant). Such an asymmetry is probably to be expected since the dilution is an action of only harvesting the population, and the predator being advantaged aids the harvesting, whereas the predator being disadvantaged hampers the harvesting. 
%This has consequences, as we shall see, for achievable performance, and even for global stabilization,  

In the uncontrolled case, $u=u^*$, the resulting system~\eqref{eq:ode-system} becomes
\begin{subequations}
\label{eq:ode-system-uncontrolled}
    \begin{align}
        \dot{\eta}_1 &= -\phi_2(\eta_2) \\
        \dot{\eta}_2 &= \phi_1(\eta_1)
    \end{align}
\end{subequations}
and can be written as
\begin{equation}
    (1-e^{-\eta_1}) d\eta_1 = \lambda_1 \lambda_2 (1-e^{\eta_2}) d\eta_2\,,
\end{equation}
which gives that the quantity
\begin{equation}
\label{Lyapunov-uncontrolled}
    V_0(\eta) = \Phi_1(\eta_1) +  \Phi_2(\eta_2)
\end{equation}
is conserved, namely, $\dot V_0 = 0$, and the solutions $\eta(t)$ are concentric orbits in the plane and satisfy
\begin{align}
    ( e^{-\eta_1(t)} - 1+\eta_1(t)) + \lambda_1\lambda_2 (e^{\eta_2(t)} - 1-\eta_2(t)) \nonumber \\
    = ( e^{-\eta_1(0)} - 1+\eta_1(0)) + \lambda_1\lambda_2 (e^{\eta_2(0)} - 1-\eta_2(0))\,.\label{eq:eta-orbits}
    %\mbox{const.}
\end{align}
However, while $V_0$ is a Lyapunov function in the uncontrolled case, it is not a control Lyapunov function. We see that from the fact that, for \eqref{eq:ode-system}, the derivative of $V_0$ is $\dot V_0 = (\phi_1(\eta_1) \reviewed{+} \phi_2(\eta_2)) (u^*-u)$, which is zero for all $\eta$ for which $\phi_1(\eta_1)= \phi_2(\eta_2)$, namely, for all the prey-predator states $(\eta_1,\eta_2)$ on the curve
%for all the values $\eta_1<0$ of the prey state where there is prey deficit, and for all values of the predator state $\eta_2>0$ where there is a predator surplus given by
\begin{equation}
\label{predator-surplus}
    \eta_2 = \ln\left(1+\frac{e^{-\eta_1} - 1}{\lambda_1\lambda_2}\right)\,,
\end{equation}
which passes through the origin $\eta=0$.
%namely for all states where the prey $\eta_1$ is depleted and there is a predator surplus $\eta_2$ given by \eqref{predator-surplus}. 

Before proceeding to control design, let us note that the Jacobian of \eqref{eq:ode-system-uncontrolled} is $\begin{bmatrix} 0 & -\lambda_2\\ \frac{1}{\lambda_1} & 0 \end{bmatrix}$. As the equilibrium harvesting $u^*$ grows, in accordance with Remark \ref{rem:steady}, we know that the prey thrive, while the predators are diminished, whereas from \eqref{eq:def_lambda} and \eqref{eq:ic_constraint} we know that both $\lambda_2$, the negated sensitivity of prey to predators, and $1/\lambda_1$, the sensitivity of predators to prey, decrease as $u^*$ grows. The eigenvalues of the Jacobian are $\pm j \frac{\lambda_2}{\lambda_1} = \pm j \sqrt{(\zeta_1 - u^*)(\zeta_2 - u^*)}$. Hence, as the equilibrium dilution/harvesting $u^*$ increases, the mutual sensitivities decrease and, as a result, the oscillations slow down. 
%Recall from Remark \ref{rem:steady} that this slowing down of oscillations is around equilibria in which the prey thrive, while the predators are diminished. 
%It is perhaps natural that the cycling slows when there are, on the average, fewer predators consuming more prey. 

\subsection{A CLF Feedback Design}
For the purpose of stabilization, instead of the Lyapunov function \eqref{Lyapunov-uncontrolled}, we propose the (positive definite and radially unbounded) {\em control Lyapunov function} candidate
\begin{align}
    V_1(\eta)=\Phi_1(\eta_1) + (1+\varepsilon) \Phi_2(\eta_2),
\label{eq:lyapunov-ode}
\end{align}
with a positive weight $(1+\varepsilon) \neq 0$ to be determined. The question is, shall we prioritize the predator in this Lyapunov function, by taking $\varepsilon>0$, or the prey, by taking $\varepsilon\in(-1,0)$? The answer to this question is easy, by computing $\dot V_1 = \left(\phi_2+(1+\varepsilon)\phi_1\right) (u^*-u) + \varepsilon \phi_1\phi_2$ and noting than $\dot V_1 = \reviewed{\varepsilon \phi_1\phi_2} = - \varepsilon (1+\varepsilon) \phi_2^2 < 0$ whenever $\phi_1 = -(1+\varepsilon) \phi_2$ and $\varepsilon>0$. 

Hence, we take $\varepsilon>0$ in \eqref{eq:lyapunov-ode}, and along with it, the CLF-based control law 
\begin{align}
    \fbox{$u = u^* + \beta \big(\phi_1(\eta_1) + {(1+\varepsilon)}\phi_2(\eta_2)\big) $}
    \label{eq:control_law}
\end{align}
which is  a gradient (``$L_gV$'') feedback relative to $V_1$ with a positive gain $\beta$, namely, $u = u^* - \beta \left(\frac{\partial V_1}{\partial\eta} \left[\begin{array}{c} -1 \\ -1 \end{array} \right]\right)^{T}$,  and much simpler than the linearizing feedback \eqref{eq-fbk-lin}. The feedback \eqref{eq:control_law} results in the Lyapunov derivative
\begin{align}
    \dot{V}_1(\eta) 
    &= -
    \begin{bmatrix}
        \phi_1 &\phi_2
    \end{bmatrix} 
Q
% \underbrace{\begin{bmatrix}
    %      \beta & \frac{\varepsilon- 2 \beta(1+\varepsilon)}{2}\\
    %     \frac{\varepsilon- 2 \beta(1+\varepsilon)}{2}  & \beta (1+\varepsilon)^2
    % \end{bmatrix}}_{Q}
\begin{bmatrix}
        \phi_1\\ \phi_2
    \end{bmatrix}
    \,.
    \label{eq:definition_q_matrix}
\end{align}
where
\begin{equation}
\label{matrix-Q}
    Q = \begin{bmatrix}
         \beta & \frac{\varepsilon- 2 \beta(1+\varepsilon)}{2}\\
        \frac{\varepsilon- 2 \beta(1+\varepsilon)}{2}  & \beta (1+\varepsilon)^2
    \end{bmatrix}
\end{equation}
%The matrix $Q$ 
is a positive definite matrix for 
\begin{align}
    \varepsilon &> 0 , \label{eq:lower_bound_q}\\
    \beta &> \beta^*(\varepsilon) = \frac{\varepsilon}{4(1+\varepsilon)} \label{eq:lower_bound_beta}\,,
\end{align}
%where $\beta^*(\varepsilon)$ is a positive hyperbola for $\varepsilon>0$. 
and its smaller eigenvalue in that case is
\begin{align}
&\lambda_\mathrm{min}(Q) = \nonumber \\
&\frac{\varepsilon}{2}
\frac{4(1+\varepsilon)\beta-\varepsilon}{\beta(1+(1+\varepsilon)^2) + \sqrt{\beta^2(1+(1+\varepsilon)^2)^2 -\varepsilon (4(1+\varepsilon)\beta-\varepsilon)}} \nonumber \\ &>0\,.\label{eq:lambda_min}
\end{align}
% \carina{\begin{align*}
%     \lambda_{1,2} = \frac{1}{2}\left( \beta \left(1+(1+\varepsilon)^2\right) \pm \sqrt{\beta^2 \left(1+(1+\varepsilon)^2\right)^2 - \varepsilon\left(4\beta(1+\varepsilon)-\varepsilon\right)} \right)
% \end{align*}}
The square root in the denominator is real since its argument is no smaller than $\left(\varepsilon\frac{(1+\varepsilon)^2-1}{(1+\varepsilon)^2+1}\right)^2>0$. Furthermore, if we take, for example, $\beta = \frac{\varepsilon}{2(1+\varepsilon)}$, we get 
simply $\lambda_\mathrm{min}(Q) = \frac{\varepsilon(1+\varepsilon)}{2}$. In the sequel, the quantity $\frac{1}{\gamma_\circ}=\frac{2  \lambda_\mathrm{min}(Q)}{1+\varepsilon}$ will arise in expressions like \eqref{gamma-circ}, \eqref{eq:constraints}, which for this particular choice of $\beta$ becomes simply $\frac{1}{\gamma_\circ} 
%=\frac{2  \lambda_\mathrm{min}(Q)}{1+\varepsilon} 
= \varepsilon$.

Hence, when picking $\varepsilon$ and $\beta$ in accordance with conditions (\ref{eq:lower_bound_q}) and (\ref{eq:lower_bound_beta}), the control law (\ref{eq:control_law}), resulting in the closed loop
\begin{subequations}
    \begin{align}
        \dot{\eta}_1 &= - \beta\phi_1(\eta_1)-(1+\beta {(1+\varepsilon)})\phi_2(\eta_2)   \\
        \dot{\eta}_2 &= -\beta {(1+\varepsilon)}\phi_2(\eta_2) + (1 - \beta) \phi_1(\eta_1),
    \end{align}
\end{subequations}
globally asymptotically stabilizes the origin $\eta=0$ of the ODE system (\ref{eq:ode-system}), in which the $\psi$-dynamics are neglected.

\begin{theorem} \label{thm:global-stabilization-controller-0}
Under the feedback law \eqref{eq:control_law}, the equilibrium $\eta=0$ of the system \eqref{eq:ode-system} is globally asymptotically and locally exponentially stable, while the control signal $u(t)$ remains bounded though not necessarily positive. Furthermore, $u(t)>0$ for all $t\geq 0$ and
\reviewed{for all $\eta(0)$  belonging to a level set $\Omega=\{ \eta \in \mathbb{R}^2 : V_1(\eta)\leq r\}$ for some $r>0$ for which $\Omega$ is a subset of
%for all $\eta(0)$ belonging to the largest level set of $V_1(\eta)$  within 
$\mathcal{D}_0 = \left\{ \eta \in \mathbb{R}^2  \Big|     u^* + \beta (\phi_1(\eta_1) + (1+\varepsilon)\phi_2(\eta_2)) > 0 \right\}$}.
\end{theorem}

The choice \eqref{eq:lower_bound_q} of $\varepsilon>0$ arises mathematically, but there is also intuition behind it. It makes sense to prioritize the predator $\eta_2$ in the CLF \eqref{eq:lyapunov-ode} because both the dilution control and the predator are {\em harvesters}: for values of $(\eta_1,\eta_2)$ for which the dilution harvesting (of both populations) is unable to affect $\dot V_1$, the predator's harvesting (or prey) is already driving $V_1$ towards zero. 

With \eqref{eq:ode-system-before-phi}, note that the dilution feedback $u(\eta_1,\eta_2)$ in \eqref{eq:definition_q_matrix} is an increasing function of both the prey biomass $\eta_1$ and the predator biomass $\eta_1$. Not favoring either of the two populations makes sense for a controller whose objective is stabilization of such a two-species system in which the natural behavior is a threat to both species.

\section{Stability with Nonzero $\psi$} \label{sec:stability-nonzero-psi}
In the following, we show that the control law (\ref{eq:control_law}) based on the ODE system neglecting the internal dynamics (\ref{eq:ode-system}) stabilizes the full ODE-IDE system (\ref{eq:transformed-system}). First, we introduce the mapping $v_i:\mathcal{S} \rightarrow \mathbb{R}_+$,
\begin{align}
    v_i(\psitminusa{i}) = \ln \left(1 + \int_0^A \bar{g}_j(a) \psi_i(t-a) \diff a\right),
    \label{eq:definition-v-map} %
\end{align}
for $i,j \in \{1,2\}$, $i\neq j$, and 
\begin{align}
    \bar{g}_i(a) & = \frac{g_i(a)\xstar{j}}{\int_0^A g_i(a)\xstar{j} \diff a}, \ \int_0^A \bar{g}_i(a) \diff a = 1.\label{eq:definition-g-bar}
\end{align}
Then (\ref{eq:doteta1}), (\ref{eq:doteta2}) are rewritten as
\begin{subequations}
\begin{align} 
    %\dot{\eta}_i = u^* - u + \phi_j(\eta_j+v_j(\psitminusa{j})).
\label{eq-eta1-perturbed}
   \dot{\eta}_1 &= u^* - u - \phi_2(\eta_2+v_2(\psi_2)),\\
\label{eq-eta2-perturbed}
   \dot{\eta}_2 &= u^* - u + \phi_1(\eta_1+v_1(\psi_1)).
\end{align}
\end{subequations}
With the previously defined control law (\ref{eq:control_law}), the closed loop is
\begin{subequations}
\label{eq:ode-ide-closed-loop}
\begin{align}
    \dot{\eta}_1 &= -\beta\left(\phi_1(\eta_1) + (1+\varepsilon) \phi_2(\eta_2)\right) - \phi_2(\eta_2+v_2(\psi_2)),\\
    \dot{\eta}_2 &= -\beta\left(\phi_1(\eta_1) + (1+\varepsilon) \phi_2(\eta_2)\right) + \phi_1(\eta_1+v_1(\psi_1)),\\
    %\dot{\eta}_i &= -\beta\left(\phi_1(\eta_1) - (1+\varepsilon) \phi_2(\eta_2)\right) + \phi_j\left(\eta_j+v_j(\psitminusa{j})\right) \\
\label{eq:psi_i}
   \psi_i(t) &= \int_0^A  \tilde{k}_i(a) \psi_i(t-a) \diff a,
   \end{align}
\end{subequations}
%with
%\begin{align}
%    \tilde{k}_i(a):= k_i(\alpha) e^{- \int_0^a \left( \mu_i(s) + \zeta_i \right) \diff s},
%\end{align}
which displays the perturbating character of the internal dynamics $\psi_i$.
For better readability in the following, we denote
\begin{alignat}{2}
\label{psihat-defs}
    \phi_i &:= \phi_i(\eta_i), \ \ \ \hat{\phi}_i &:= \phi_i(\eta_i+v_i).
\end{alignat}
Additionally, we need a technical assumption on the birth kernel for the definition of a Lyapunov function $G$~\eqref{eq:g-functional} as used in \cite{karafyllis2017stability}.
%It is satisfied for by arbitrary mortality rates $\mu_i$ for each birth rate $k_i$ with a finite number of zeros on $[0,A]$ and states that members of the population reproduce at almost every age from zero to $A$.
\begin{assumption}[Birth kernel \cite{karafyllis2017stability}]\label{assump:technical_assumption}
There exist constants $\kappa_i>0$ such that
$\int_0^A |\tilde{k}_i(a)- z_i \kappa_i \int_a^A \tilde{k}_i(s) \diff s | \diff a < 1$ with $\reviewed{z_i = (\int_0^A a \tilde{k}_i(a)  \diff a)^{-1}}$. 
\end{assumption}
\reviewed{Assumption~\ref{assump:technical_assumption} is a mild technical assumption, since it is satisfied by arbitrary mortality rate $\mu_i$ for every birth kernel $k_i$ that has a finite number of zeros on $[0, A]$.  
The role of Assumption~\ref{assump:technical_assumption} is crucial for the establishment of the function $G$ used in the CLF. 
Means of verifying the validity of Assumption~1 and detailed discussions are given in \cite{karafyllis2017stability}.}

\reviewed{If Assumption 1 holds then there exist constants $\sigma_i>0$, with $i=1,2$, such that the inequalities $\int_0^A | \tilde{k}_i - z_i \kappa_i \int_a^A \tilde{k}_i(s) \diff s | e^{\sigma_i a} \diff a <1 $ for $i=1,2$ hold.}

Now, to state the main result, we introduce the functionals $G_i$ defined as
\begin{align}
        G_i(\psi_{i}) &: = \frac{\max_{a \in [0,A]} |\psi_i(-a)| e^{\sigma_i ({ A} - a)}}{1 + \min(0, \min_{a \in [0,A]} \psi_i(-a))} %\carina{e^{\sigma_i A}}
        , \label{eq:g-functional}
\end{align}
and recall from \cite{karafyllis2017stability} that their Dini derivatives ($D^+$) satisfy %\carina{(still valid?)}
    \begin{align}
D^+ G_i(\psitminusa{i}) \leq -\sigma_i G_i(\psitminusa{i}) 
%D^+ G_i(\psitminusa{i})^2 \leq -2\sigma_i G_i(\psitminusa{i})^2 
\label{eq:dini_g}
    \end{align}
    along solutions $\psitminusa{i}$ of the IDE for with sufficiently small parameters $\sigma_i>0$. This property follows from Corollary~4.6 and the proof of Lemma~5.1 of \cite{karafyllis2017stability} with 
%    \begin{align}
        $C_i(\psi_{i})=\frac{1}{(1 + \min(0, \min_{a \in [0,A]} \psi_i(-a)))^2}$ and $
        b(s) = s $, %\frac{1}{2} s^2$. 
%    \end{align}
    similar to what is discussed in \cite{haacker2024stabilization}. 
    %, under assumption~\ref{assump:technical_assumption} and , to obtain the functional $G_i(\psitminusa{i})$ as in (\ref{eq:g-functional}), 

\begin{theorem} \label{thm:local-stability}
Let Assumption~\ref{assump:technical_assumption} hold and define $\mathcal{S}=\mathcal{S}_1 \times \mathcal{S}_2$. Consider the closed-loop system (\ref{eq:ode-ide-closed-loop}), i.e., system (\ref{eq:transformed-system}) with the control law (\ref{eq:control_law}), 
on the state space $\mathbb{R}^2 \times S$, which is a subset of the Banach space $\mathbb{R}^2 \times C^0 ([-A,0] ; \mathbb{R}^2)$ with the standard topology. 
Suppose the parameters $\varepsilon$, $\beta$ satisfy the conditions (\ref{eq:lower_bound_q}), (\ref{eq:lower_bound_beta}) and $\lambda_\mathrm{min}(Q) > 0$ is the lowest eigenvalue of the resulting positive definite matrix $Q(\varepsilon,\beta)$ from (\ref{eq:lambda_min}).
Denote $\eta=[\eta_1, \eta_2]^\top$, $\psi = [\psi_{1}, \psi_{2}]^\top$ and the Lyapunov functional 
\begin{align}
        V(\eta, \psi) &= 
        V_1(\eta) + \frac{\gamma_1 }{\sigma_1}
        h(G_1(\psi_{1})) 
        + \frac{\gamma_2 }{\sigma_2}
        %e^{\sigma_2 A}} 
        %\gamma_2 
        h(G_2(\psi_{2})) 
        \label{eq:lyapunov-ode-ide}
\end{align}
with $V_1(\eta)$ from (\ref{eq:lyapunov-ode}),  the positive weights chosen as
\begin{align}
\label{gamma-circ}
    & \gamma_1 > \frac{1}{\lambda_1^2} \gamma_\circ, \quad \gamma_2 > \lambda_2^2\gamma_\circ, 
    \quad \gamma_\circ = 
        \frac{1+\varepsilon }{2 \lambda_\mathrm{min}(Q)}\,,
\end{align}
and the positive definite radially unbounded function $h(\cdot)$ defined as
\begin{align}
\label{h-Lyapunov}
       h(p)&:= %\frac{1}{\sigma} 
       \int_0^p \frac{1}{z} \left(e^z-1\right)^2 \diff z.
\end{align}
Then, 
\begin{enumerate}
    \item
positive invariance holds for all the level sets of $V$ of the form
\begin{align}
\label{eq:Omegac}
    \Omega_c:=\{\eta \in \mathbb{R}^2, \psi \in \mathcal{S} \ |  \ V(\eta, \psi)\leq c\}
    %\subset \mathcal{D}
\end{align} 
that are within the set 
\begin{align}
\label{eq:constraints}
    \mathcal{D} := &\Biggl\{ \eta \in \mathbb{R}^2, \psi \in \mathcal{S} \ 
       \Bigg| \nonumber\\
       & \ \eta_1\geq-\ln\left(\lambda_1\sqrt{ 
       \frac{\gamma_1}{\gamma_\circ}
    %\frac{2 \gamma_1 \lambda_\mathrm{min}(Q)}{1+\varepsilon}
    }\right), 
      % \nonumber\\
      \ \ 
    \eta_2\leq \ln\left(\frac{1}{\lambda_2}\sqrt{
    \frac{\gamma_2}{\gamma_\circ}
    %\frac{2 \gamma_2 \lambda_\mathrm{min}(Q)}{1+\varepsilon}
    }\right), 
    \nonumber\\
    &  \ \ %u= 
    u^* + \beta (\phi_1(\eta_1) + (1+\varepsilon)\phi_2(\eta_2)) >
 0 \Biggr\},
\end{align}
namely, within the set $\mathcal{D}$ where neither the prey deficit nor the predator surplus are  too big but the predator state is bigger than $\eta_2(\eta_1) = \ln\left(1+\frac{e^{-\eta_1} - 1 - \frac{\lambda_1}{\lambda_2}\frac{u^*}{\beta (1+\varepsilon)}}{(1+\varepsilon)\lambda_1\lambda_2}\right)$; 
\item the input $u(t)$ remains positive for all time; 
%\item the equilibrium point $(\eta, \psi)=0$ is stable, and 
\item there exists $\theta_0\in\mathcal{KL}$ such that, for all initial conditions $(\eta_0, \psi_0)$ within the largest $\Omega_c$ contained in $\mathcal{D}$, the following estimate holds, 
%states $\eta$, $\psi$ converge to zero.
\begin{equation}
\label{eq:KL-estimate}
    \left|(\eta(t),G(t)) \right| \leq \theta_0\left(\left|(\eta(0),G(0) \right|, t\right), \quad\forall t\geq 0\,,
\end{equation}
where $G(t) = \left(G_1(\psitminusa{1}),G_2(\psitminusa{2}) \right)$;
\item the equilibrium $\eta=0, \psi=0$ is locally exponentially stable in the norm $\sqrt{\eta_1^2+\eta_2^2} + \|\psi_1\|_\infty+ \|\psi_2\|_\infty$.
\end{enumerate}
% Equivalently, in the original concentration variables $(x_1,x_2)$, 
% \begin{align}
% &\left(\left|
% \left(\Pi_1[x_{1}](t)\right),
% \mathrm{ln}\left(\Pi_2[x_{2}](t)\right),
% \right.\right. \nonumber\\
% & \left.\left.
% \frac{\max_{a \in [0,A]} \left| \frac{x_{1}(a,t)}{x^*_1(a)\Pi_1[x_{1}](t)}-1\right| e^{-\sigma_1 a}}{1 + \min\left(0, \min_{a \in [0,A]}  \frac{x_{1}(a,t)}{x^*_1(a)\Pi_1[x_{1}](t)}-1\right)}, 
% \frac{\max_{a \in [0,A]} \left| \frac{x_{2}(a,t)}{x^*_2(a)\Pi_2[x_{2}](t)}-1\right| e^{-\sigma_2 a}}{1 + \min\left(0, \min_{a \in [0,A]}  \frac{x_{2}(a,t)}{x^*_2(a)\Pi_2[x_{2}](t)}-1\right)}
% \right|\right)
% \nonumber\\
% &
% \leq 
% \nonumber\\
% &
% \theta\left(\left|
% \mathrm{ln}\left(\Pi_1[x_{1}](0)\right),
% \mathrm{ln}\left(\Pi_2[x_{2}](0)\right),
% \right.\right. \nonumber\\
% & \left.\left.
% \frac{\max_{a \in [0,A]} \left| \frac{x_{1}(a,0)}{x^*_1(a)\Pi_1[x_{1}](0)}-1\right| e^{-\sigma_1 a}}{1 + \min\left(0, \min_{a \in [0,A]}  \frac{x_{1}(a,0)}{x^*_1(a)\Pi_1[x_{1}](0)}-1\right)}, 
% \frac{\max_{a \in [0,A]} \left| \frac{x_{2}(a,0)}{x^*_2(a)\Pi_2[x_{2}](0)}-1\right| e^{-\sigma_2 a}}{1 + \min\left(0, \min_{a \in [0,A]}  \frac{x_{2}(a,0)}{x^*_2(a)\Pi_2[x_{2}](0)}-1\right)}
% \right|,t\right)
% \end{align}
% for all $t\geq 0$ and all initial conditions $(x_1(\cdot,0), x_2(\cdot,0))$ %$[\eta_0, \psi_0]$ 
% within the largest $\Omega_c$ contained in $\mathcal{D}$.
\end{theorem}

To illustrate this complex interconnection of multiple constraints of $\mathcal{D}$, Figure~\ref{fig:roa-for-simulation} shows level sets of $V_1$ in the $\eta_1$-$\eta_2$-plane, together with the domain $\mathcal{D}$ in gray, i.e., the intersection of the two half planes $\eta_1>-H_1$, $\eta_2<H_2$  and the constraint $u>0$ for the ODE-system (\ref{eq:ode-system}).
The choice of $\varepsilon$ changes the shape of the level sets, whereas the shape of the boundary $u>0$ is influenced by both $\beta$ and $\varepsilon$.

\begin{figure}
    \centering
    \includegraphics{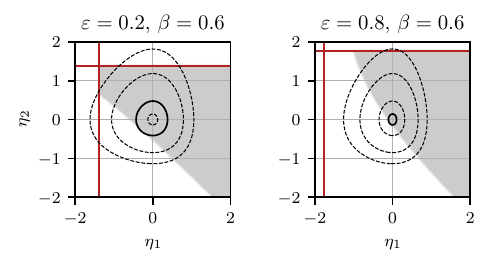}
    \caption{Level sets of $V_1$ (\mlLineLegendDashed{black}) for two different different choices of $\varepsilon$, $\beta$.
    with bounds $H_1$, $H_2$ (\mlLineLegend{rot}) for $\eta_1$, $\eta_2$. The gray area is the set $\mathcal{D}$, namely, the set in which  $\eta_1>-H_1$, $\eta_2<H_2$, and $u>0$. 
    The largest level set of $V_1$ (\mlLineLegend{black}) is contained within the actual region of attraction of $\eta=0$, for the case $\psi=0$. 
    The left plot indicates the results with the control parameters $\varepsilon$, $\beta$ as used in the simulation shown in Section~\ref{sec:simulations}.}
    \label{fig:roa-for-simulation}
\end{figure}

\mycomment{
 Trajectories of the uncontrolled system in which $\psi\equiv 0$ (\ref{eq:ode-system-uncontrolled}) are shown in Figure~\ref{fig:eta-orbits-uncontrolled} for different initial conditions $\eta_0$. The orbits are concentric in the $\eta_1$-$\eta_2$-plane and satisfy (\ref{eq:eta-orbits}). 
\begin{figure}
    \centering
    \includegraphics[]{fig/open-loop-eta-orbits-overpopulated-prey.pdf}
    \caption{Trajectories in the $\eta_1$-$\eta_2$-plane for the uncontrolled system with $u=u^*$ for different initial conditions $\eta_0$ that represent an overpopulated prey. }
    \label{fig:eta-orbits-uncontrolled}
\end{figure}
}

While we design the control law as (\ref{eq:control_law}) and offer its alternative representation as \eqref{eq:control-in-x}, this representation requires  the age-specific measurement of $x_i(a,t)$.
%and the knowledge of the birth kernels $k_i(a)$. This is not realistic. 
Since the controller (\ref{eq:control_law}) uses only $e^{\eta_2}$ and $e^{-\eta_1}$, we recall that, if measurements are available of 
\begin{equation}
    y_i(t) = \int_0^A c_i(a) x_i(a,t) da\,, \quad i=1,2
\end{equation}
with sensor kernels $c_1(a)$ and $c_2(a)$, then we can use the fact that
\begin{equation}
    e^{\eta_i(t)} = \frac{y_i(t)}{y_i^*}\frac{1}{1+\int_0^A p_i(\alpha) \psi_i(t-\alpha)d\alpha}
\end{equation}
where $y_i^* = \int_0^A c_i(a) x_i^*(a) da$ and $p_i(a) = \frac{c_i(a)x_i^*(a)}
{y_i^*}%{\int_0^A c_i(a) x_i^*(a) da}
$, and the fact that $\psi_i(t)$ decays exponentially, in accordance with \eqref{eq:psi-decays}, to approximate $e^{\eta_i(t)}$ as $e^{\eta_i(t)}\approx  \frac{y_i(t)}{y_i^*}$ and approximate the control law  (\ref{eq:control_law}) as
\begin{align}
\label{eq:control-in-y}
    u = & u^* + \beta \left[
    \frac{1}{\lambda_1}
    \left(1 - \frac{y_1^*}{y_1(t)} \right)
 %   \right. \nonumber\\ \left. &
    - (1+\varepsilon)\lambda_2\left(1 - \frac{y_2(t)}{y_2^*}\right)
    \right]\,.
\end{align}
It should be noted though that $y_1^*$ and $y_2^*$ in \eqref{eq:control-in-y} depend on $u^*$ in the following manner,
\begin{align}
y_1^*(u^*) = & \frac{1}{\zeta_2 - u^*} \frac{\int_0^A c_1(a) e^{-\int_0^a\left(\zeta_1+\mu_1(s)\right) \diff s} \diff a}{\int_0^A g_2(a) e^{-\int_0^a\left(\zeta_1+\mu_1(s)\right) \diff s} \diff a}\\
y_2^*(u^*) = & \left(\zeta_1 - u^*\right) \frac{\int_0^A c_2(a) e^{-\int_0^a\left(\zeta_2+\mu_2(s)\right) \diff s} \diff a}{\int_0^A g_1(a) e^{-\int_0^a\left(\zeta_2+\mu_2(s)\right) \diff s} \diff a}\,, 
\end{align}
and require the knowledge of the model functions $\mu_i, k_i, g_i$ and the sensor functions $c_i$. 
We don't prove a theorem under the controller \eqref{eq:control-in-y} because neglecting the decaying term $\int_0^A p_i(\alpha) \psi_i(t-\alpha)d\alpha$ introduces an additional perturbation in the closed-loop system.

\begin{proof}[Proof of Theorem \ref{thm:local-stability}.]
    The existence and uniqueness of the solution $\psitminusa{i} \in \mathcal{S}$ and the property
    \begin{align}
        \inf_{t\geq A} \psi_i(t) \geq \min_{t \in [-A, 0]} \psi_i(t) > -1 \ \forall t \geq 0. \label{eq:psi_lower_bound}
    \end{align}
    is provided in Lemma~4.1 of \cite{karafyllis2017stability}. Further, the initial condition is lower bounded and given by
    \begin{align}
        \psi_{i,0}(a) := \psi_i(-a) = \frac{x_{i,0}(a)}{\xstar{i}\Pi[x_{i,0}]} - 1 > -1,
    \end{align}
    where $x_{i,0}(a)$, $\xstar{i}$, and $\Pi[x_{i,0}]$ from (\ref{eq:pi-function}) are positive on the domain $[0,A]$. With the $\bar{g}_i(a) \geq 0$, $\int_0^A \bar{g}_i(a)\diff a=1$ by definition (\ref{eq:definition-g-bar}), the map $v_i(\psitminusa{i})$ (\ref{eq:definition-v-map}) is well-defined and continuous. Hence, the ODE subsystem of (\ref{eq:ode-ide-closed-loop}) locally admits a unique solution.

    % Definition of Lyapunov Function
Furthermore, from \cite{karafyllis2017stability}, more precisely (A.43), we know that
    \begin{align}
    |v_i(\psitminusa{i})| \leq G_i(\psitminusa{i}) 
    %e^{\sigma_i A}
    \label{eq:v_leq_G}
    \end{align}
   where $v_i:\mathcal{S} \rightarrow \mathbb{R}_+$ is defined in \eqref{eq:definition-v-map}, holds.

    The derivative of $V_1(\eta)$ can be bounded to
    \begin{align}
        \dot{V}_1 
        =&  -\begin{bmatrix}
            \phi_1 \\ \phi_2
        \end{bmatrix}
        Q
        \begin{bmatrix}
            \phi_1 &\phi_2
        \end{bmatrix} 
            + \phi_1\phi_2 - \phi_1 \hat{\phi}_2 \nonumber \\
            &- \left(1+\varepsilon\right)\phi_1\phi_2 + \left(1+\varepsilon\right) \hat{\phi}_1 \phi_2
            \nonumber \\
                 \leq& -\lambda_\mathrm{min}(Q) ||\phi||^2 +
    %\begin{bmatrix}
    %    \phi_1 & q \phi_2
    %\end{bmatrix}
    %\begin{bmatrix}
    %    \hat{\phi}_1 - \phi_1 \\
    %     \hat{\phi}_2 - \phi_2
    %\end{bmatrix} 
    \begin{bmatrix}
        -\phi_1 & (1+\varepsilon) \phi_2
    \end{bmatrix}
    \begin{bmatrix}
        \hat{\phi}_2 - \phi_2 \\
         \hat{\phi}_1 - \phi_1
    \end{bmatrix} 
    \nonumber \ \\
    \leq& -\lambda_\mathrm{min}(Q) ||\phi||^2 + (1+\varepsilon) ||\phi|| \ ||\hat{\phi}-\phi|| 
    \nonumber\\
    \leq& -\lambda_\mathrm{min}(Q) ||\phi||^2
    + \frac{(1+\varepsilon)}{2\varsigma} ||\phi||^2 + \frac{\varsigma}{2} ||\hat{\phi}-\phi||^2 
    \nonumber\\
            \leq& -\frac{\lambda_\mathrm{min}(Q)}{2} ||\phi||^2 
    + \frac{(1+\varepsilon)}{2\lambda_\mathrm{min}(Q)} ||\hat{\phi} - \phi||^2
    \label{eq:vdot_perturbed}
    \end{align}
    with the positive definite matrix $Q$ from (\ref{eq:definition_q_matrix}) for certain $\varepsilon$, $\beta$ as stated in (\ref{eq:lower_bound_q}) and (\ref{eq:lower_bound_beta}), and its lowest real, positive eigenvalue $\lambda_\mathrm{min}(Q)$, as well as using the fact that $\varepsilon>0$ and Young's inequality with $\varsigma=(1+\varepsilon) \lambda_\mathrm{min}(Q)^{-1}$.
    %Combining this results with i
    Inserting the definitions of $\phi_i$ into $\hat{\phi}_i-\phi_i$, we get
    \begin{subequations}
    \label{phitildes-fulldef}
       \begin{align}
               \hat{\phi}_1-\phi_1 &=  
        {\left( \phi_1-\frac{1}{\lambda_1} \right)}{(e^{-v_1}-1)},\\
        \hat{\phi}_2-\phi_2  &= 
        {(\phi_2+\lambda_2 )} {(e^{v_2}-1)}\,.
        \end{align} 
    \end{subequations}
Substituting these expressions into the inequality (\ref{eq:vdot_perturbed}), with the help of \eqref{h-Lyapunov} and \eqref{eq:v_leq_G} the Dini derivative of $V(\eta, \psi$) defined in (\ref{eq:lyapunov-ode-ide}), results in 
\begin{align}
        &D^+ {V} \nonumber \\
        &\leq -\frac{\lambda_\mathrm{min}(Q)}{2} ||\phi||^2 
        + \frac{1+\varepsilon}{2\lambda_\mathrm{min}(Q)} ||\hat{\phi}-\phi||^2 \nonumber \\
        &+ \frac{\gamma_1 }{\sigma_1} D^+ h(G_1)+ \frac{\gamma_2}{\sigma_2}
        D^+ h(G_2)
        \nonumber\\
        &\leq -\frac{\lambda_\mathrm{min}(Q)}{2} ||\phi||^2 
        + \frac{1+\varepsilon}{2\lambda_\mathrm{min}(Q)} \nonumber \\
        &\Big( \left(\lambda_2 + \phi_2\right)^2\left(e^{v_2}-1\right)^2 
        + \left(\frac{1}{\lambda_1} - \phi_1\right)^2\left(e^{-v_1}-1\right)^2
        \Big)  \nonumber\\
        &- \gamma_1  \left(e^{G_1}-1\right)^2 
        -  \gamma_2  \left(e^{G_2}-1\right)^2  
        \nonumber\\
        &\leq  -\frac{\lambda_\mathrm{min}(Q)}{2} ||\phi||^2 \nonumber \\
        &+ \left[\frac{1+\varepsilon}{2\lambda_\mathrm{min}(Q)} \left(\frac{1}{\lambda_1} - \phi_1\right)^2 -\gamma_1  \right] \left(e^{G_1}-1\right)^2
         \nonumber\\
        &
        + \left[\frac{1+\varepsilon}{2\lambda_\mathrm{min}(Q)} 
        \left(\lambda_2 + \phi_2\right)^2  - \gamma_2  \right] \left(e^{G_2}-1\right)^2. \label{eq:temp-mk}
    \end{align} 
Consider for a moment the expressions $\left(\frac{1}{\lambda_1} - \phi_1\right)^2$ and $\left(\lambda_2 + \phi_2\right)^2$ above. Note that they cannot be bounded, respectively, in direct proportion to $\Phi_1$ and $\Phi_2$, since
\begin{equation}
 \lim_{\eta_1\rightarrow -\infty} \frac{\frac{1}{\lambda_1}+\Phi_1(\eta_1)}{\frac{1}{\lambda_1} -\phi_1(\eta_1)}
 = \lim_{\eta_2\rightarrow +\infty} \frac{\lambda_2 +\Phi_2(\eta_2)}{\lambda_2 +\phi_2(\eta_2)}
= 1\,,
\end{equation}
but only quadratically in $\Phi_1$ and $\Phi_2$. This means that one cannot attain a global result by replacing $V_1$ with $\ln(1+V_1)$ in \eqref{eq:lyapunov-ode-ide}. 
So, for \eqref{eq:temp-mk} to be negative definite at least in a region of the state space around the origin,
%and achieve $\dot{V}<0$, 
we restrict the states $\eta$
\begin{subequations} \label{eq:state-constraints}
\begin{align}
    \eta_1 &\geq -\ln\left(\lambda_1\sqrt{\frac{2 \gamma_1 \lambda_\mathrm{min}(Q)}{1+\varepsilon}}\right)=
    :- H_1,\\
    \eta_2 &\leq \ln\left(\frac{1}{\lambda_2}\sqrt{\frac{2 \gamma_2  \lambda_\mathrm{min}(Q)}{1+\varepsilon}}\right)=:H_2.
\end{align}
\end{subequations}
For $H_1$, $H_2$ to be positive, choose
\begin{subequations} \label{eq:state-bounds-def}
    \begin{align}
        \gamma_1 &> \frac{1}{\lambda_1^2}\frac{1+\varepsilon}{2 \lambda_\mathrm{min}(Q)},\\
        \gamma_2 &>  \lambda_2^2 \frac{1+\varepsilon}{2 \lambda_\mathrm{min}(Q)}.
    \end{align}
\end{subequations}
After imposing (\ref{eq:state-bounds-def}), (\ref{eq:temp-mk}) becomes
\begin{align}
\label{eq:Vdot-final}
    D^+ {V} \leq &-\frac{\lambda_\mathrm{min}(Q)}{2} ||\phi||^2 
     - \frac{\gamma_1}{2}  \left(e^{G_1}-1\right)^2 
        - \frac{\gamma_2}{2}\left(e^{G_2}-1\right)^2 \nonumber\\
%    < 0 
     &\forall \quad \eta_1>-H_1, \ \eta_2<H_2.
\end{align}
Since $V$ is radially unbounded, all level sets $\Omega_c:=\{\eta \in \mathbb{R}^2, \psi \in \mathcal{S} | V(\eta, \psi)\leq c\}$ are compact. Hence, there exists a level set $c>0$ such that $\Omega_c \subset \mathcal{D}$, where $\mathcal{D}$ is defined by (\ref{eq:constraints}). Let us now rewrite \eqref{eq:lyapunov-ode-ide}, \eqref{eq:Vdot-final}, \eqref{eq:lyapunov-ode} as
\begin{align}
\label{eq:Vdot-final+}
&D^+ V(\eta,G) 
\nonumber\\
& = %\frac{d}{dt} 
D^+ \left[ 
%V_1(\eta)
\Phi_1(\eta_1) + (1+\varepsilon) \Phi_2(\eta_2),
+ \frac{\gamma_1 }{\sigma_1}
%e^{\sigma_1 A}} 
h(G_1) 
        + \frac{\gamma_2 }{\sigma_2}
%e^{\sigma_2 A}} 
        h(G_2) \right]
\nonumber\\
&\leq -\frac{\lambda_\mathrm{min}(Q)}{2} ||\phi(\eta)||^2 
     - \frac{\gamma_1}{2}  \left(e^{G_1}-1\right)^2 
%    \nonumber\\ &
        - \frac{\gamma_2}{2}    \left(e^{G_2}-1\right)^2 
        \nonumber\\
        & =: - W(\eta,G), \qquad G=(G_1,G_2)
\end{align}
for all $\eta_1>-H_1,  \eta_2<H_2$. Since $V_1$ and $W$ are positive definite in $(\eta,G)$, there exists $\theta_0\in\mathcal{KL}$ such that, for all initial conditions $(\eta_0, G_0)$ within the largest $\Omega_c$ contained in $\mathcal{D}$, the estimate \eqref{eq:KL-estimate} holds. 
Finally, since \eqref{eq:lyapunov-ode} with \eqref{Phi-defs} is locally quadratic in $\eta$, and since for all $r\in(0,1)$ and for all $\psi_i\in\mathcal{S}$ with $\|\psi_i\|_\infty \leq r$ it holds, using \eqref{eq:g-functional}, that
\begin{equation}
\|\psi_i\|_\infty \leq G_i(\psi_i) \leq \frac{{\rm e}^{\sigma_i A}}{1-r} \|\psi_i\|_\infty \,,
\end{equation}
the local asymptotic stability of $\eta=0, \psi=0$ in the norm $\sqrt{\eta_1^2+\eta_2^2} + \|\psi_1\|_\infty+ \|\psi_2\|_\infty$ follows. The exponential nature of stability follows from a careful inspection of \eqref{eq:Vdot-final+}, along with the definitions of $V, \phi_1,\phi_2, G_1, G_2$. 
This completes the proof of Theorem~\ref{thm:local-stability}.
\end{proof}

\begin{corollary}
There exists $\theta\in\mathcal{KL}$ such that, under the restrictions on the initial conditions in Theorem~\ref{thm:local-stability} given by \eqref{eq:constraints}, \eqref{eq:Omegac}, \eqref{eq:lyapunov-ode-ide} but understood in the sense of the transformations $(x_1,x_2)\mapsto (\eta_1,\eta_2,\psi_1,\psi_2)$ defined by \eqref{eq:system-trafo} applied to the initial conditions $x_i(\cdot,0)\in\mathcal{F}_i$, the control law (\ref{eq:control_law}), given in the biomass concentration variables $(x_1,x_2)$ as
\begin{align}
\label{eq:control-in-x}
    u = & u^* + \beta \Bigg[
    \frac{1}{\lambda_2}
    \left(1 - \frac{\int_0^A a k_1(a) x_1^*(a) \diff a}{\int_0^A \pi_{0,1}(a) x_1(a,t)  \diff a} \right) \nonumber \\
    &- \lambda_1\left(1 - \frac{\int_0^A \pi_{0,2}(a) x_2(a,t)  \diff a}
    {\int_0^A a k_2(a) x_2^*(a) \diff a} \right)
    \Bigg]\,,
\end{align}
guarantees the following regional asymptotic stability estimate:
\begin{align}
\label{eq:final-KL-estimate}
\max_{a\in[0,A]}\left| \ln \frac{x_i(a,t)}{x_i^*(a)}\right| \leq \theta\left( \max_{a\in[0,A]}\left| \ln \frac{x_i(a,0)}{x_i^*(a)}\right|, t \right)\,, \ \forall t\geq 0\,.
\end{align}
\end{corollary}

\begin{proof}
It was established in the inequalities (5.24) and (5.27) in \cite{karafyllis2017stability} that there exist $\bar\theta, \underline\theta \in \mathcal{K}$ such that 
\begin{eqnarray}
\max_{a\in[0,A]}\left| \ln \frac{x_i(a,t)}{x_i^*(a)}\right| &\leq& \bar\theta(|(\eta(t), G(t))|)\,, \ \forall t\geq 0
\\
|(\eta(0), G(0))| &\leq & \underline\theta \left( \max_{a\in[0,A]}\left| \ln \frac{x_i(a,0)}{x_i^*(a)}\right| \right)\,.
\end{eqnarray}
Combining these two inequalities with \eqref{eq:KL-estimate}, the estimate \eqref{eq:final-KL-estimate} follows immediately, with $\theta(r,t) = \bar\theta\left(\theta_0\left(\underline\theta(r),t\right)\right)$.
\end{proof}

\section{%Global 
Global Stabilization with Positive Dilution}\label{sec:stability-positive-dilution}

Up to the present section, we focused on the global stabilization of the reduced model \eqref{eq:ode-system}, namely, for all initial conditions $\eta(0)\in\mathbb{R}^2$, 
but without ensuring the that dilution $u(t)$ remains positive for all initial conditions. Dilution values, i.e., harvesting rates, that take negative values amount to introducing (externally ``farmed'') populations, which is unrealistic, especially when such an injection of populations needs to be in proportion to the current density of both predator and prey at each respective age. 

In this section we return to the reduced model \eqref{eq:ode-system}, with $\phi_1, \phi_2$ defined in \eqref{phi-def}, and define a globally stabilizing feedback law with positive $u(t)$.

\begin{theorem} \label{thm:global-stabilization-controller}
Under the feedback law
\begin{equation}
    u = u^* + \varepsilon\phi_2(\eta_2) + \beta \frac{\varphi(\eta)}{\sqrt{\delta^2+ \left(\min(0, \varphi(\eta)\right)^2}} 
    \label{eq:global-stabilization-controller}
\end{equation}
where 
\begin{equation}
\label{gain-cond}
    \varphi(\eta) = \phi_1(\eta_1) + \left(1+\varepsilon\right) \phi_2(\eta_2)\,,
\end{equation}
$\delta>0$ is arbitrary, and, for a given dilution setpoint $u^*>0$, the feedback gains $\varepsilon >0$ and $\beta\geq 0$ are selected so that
\begin{equation}
    \varepsilon\lambda_2 + \beta < u^*\,, \label{eq:saturated-control-constraints}
\end{equation}
the origin $\eta=0$ of the system \eqref{eq:ode-system} is globally asymptotically stable, locally exponentially stable, and, furthermore, the dilution input $u(t)$ defined in \eqref{eq:global-stabilization-controller} remains positive for all $t\geq 0$. 
\end{theorem}

\begin{proof}
First, we observe from the fact that the minimum of $\frac{\varphi(\eta)}{\sqrt{\delta^2+ \left(\min(0, \varphi(\eta)\right)^2}}$ is $-1$, the minimum of $\phi_2$ in \eqref{phi-def} is $-\lambda_2$, the gain condition \eqref{eq:saturated-control-constraints}, and the definition of the feedback $u(\eta)$ in \eqref{eq:global-stabilization-controller} that $u$ remains positive. To prove global  asymptotic stabilization, we take the Lyapunov function $V_1$ defined in \eqref{eq:lyapunov-ode} with the help of \eqref{Phi-defs}, and obtain
\begin{align}
\label{eq-V1dot-beta}
    &\dot{V}_1(\eta) \nonumber \\
    =& \phi_1(\eta_1) \Bigg(-\left(1+\varepsilon\right) \phi_2(\eta_2)
    - \frac{\beta \varphi(\eta)}{\sqrt{\delta^2+ \left(\varphi^-\right)^2}}\Bigg) \nonumber \\
    &+ \left(1+\varepsilon\right) \phi_2(\eta_2) \Bigg(\phi_1(\eta_1) - \varepsilon \phi_2(\eta_2) 
    - \frac{\beta \varphi(\eta)}{\sqrt{\delta^2+ \left(\varphi^-\right)^2}}\Bigg)\nonumber\\
    =& - \varepsilon(1+\varepsilon)\phi_2^2(\eta_2)\nonumber\\
    &- \left(\phi_1(\eta_1)+(1+\varepsilon)\phi_2(\eta_2)\right)  \frac{\beta \varphi(\eta)}{\sqrt{\delta^2+ \left(\varphi^-\right)^2}} \nonumber\\
    =& - \varepsilon(1+\varepsilon)\phi_2^2(\eta_2) -  \frac{\beta \varphi^2(\eta)}{\sqrt{\delta^2+ \left(\varphi^-\right)^2}}\,,
\end{align}
\mycomment{
\carina{}
\begin{align}
    \dot{V}_1(\eta) 
    =& \phi_1(\eta_1) \left(-\left(1+\varepsilon\right) \phi_2(\eta_2)
    - \Xi(\varphi) \right) \nonumber\\
    &+ \left(1+\varepsilon\right) \phi_2(\eta_2) \left(\phi_1(\eta_1) - \varepsilon \phi_2(\eta_2) 
    - \Xi(\varphi) \right)\nonumber\\
    =& - \varepsilon(1+\varepsilon)\phi_2^2(\eta_2)
    - \left(\phi_1(\eta_1)+(1+\varepsilon)\phi_2(\eta_2)\right) \Xi(\varphi) \nonumber\\
    =& - \varepsilon(1+\varepsilon)\phi_2^2(\eta_2) -  \Xi(\varphi)\,,
\end{align}}
%Therefore, $\dot{V}_1(\eta)$ 
where we have denoted $\varphi^{-} = \min(0, \varphi)$. When $\beta>0$, \eqref{eq-V1dot-beta} is negative definite and, consequently, $0\in\mathbb{R}^2$ is globally asymptotically stable for the closed-loop system (\ref{eq:ode-system}) with (\ref{eq:global-stabilization-controller}). When $\beta=0$, the set $\dot V_1=0$ is the set $\eta_2$ and the only solution that remains in this set, for the closed-loop system (\ref{eq:ode-system}), (\ref{eq:global-stabilization-controller}), is the solution with $\eta_1(t)\equiv 0$. By the Barbashin-Krasovskii theorem, $0\in\mathbb{R}^2$ is globally asymptotically stable. The Jacobian matrix of the closed-loop system (\ref{eq:ode-system}) with (\ref{eq:global-stabilization-controller}) at $0\in\mathbb{R}^2$,
\begin{align}
\label{Jacobian-Iasson}
%    J_{V_1} = 
\begin{bmatrix}
         -\frac{k}{\lambda_1} & -(1+\varepsilon)\lambda_2(1+k) \\
        \frac{1}{\lambda_1}\left(1-k\right) & -\varepsilon\lambda_2 - k (1+\varepsilon)\lambda_2
    \end{bmatrix},
\end{align}
where $k=\frac{\beta}{\delta}$, is a Hurwitz matrix for both $\beta>0$ and $\beta= 0$. Therefore, $0 \in \mathbb{R}^2$ is also 
%LES (\carina{meaning?}) 
locally exponentially stable for the closed-loop system (\ref{eq:ode-system}) with (\ref{eq:global-stabilization-controller}). 
\end{proof}

We briefly examine the local exponential performance of the closed-loop system. The characteristic polynomial of the Jacobian \eqref{Jacobian-Iasson} is given by
\begin{eqnarray}
    p(s) &=& s^2 + \left( k \left(\frac{1}{\lambda_1} + \lambda_2\right) + \varepsilon \lambda_2 \left(1+k\right) \right) s \nonumber \\
    &&+ \left(1+\varepsilon\left(1+k\right)\right) \frac{\lambda_2}{\lambda_1}.
\end{eqnarray}
By examining the discriminant condition of this polynomial, 
\begin{equation}
\label{discriminant}
    \left( k \left(\frac{1}{\lambda_1} + \lambda_2\right)  + \varepsilon \lambda_2 \left(1+k\right) \right)^2 \geq 4 \left(1+ \varepsilon \left(1+k\right)\right) \frac{\lambda_2}{\lambda_1}
\end{equation}
it is evident that the roots of the polynomial are not only negative but also real 
%for sufficiently large $k>0$ (i.~e., for sufficiently small $\delta>0$), the eigenvalues of the Jacobian matrix of the closed-loop system (\ref{eq:ode-system}) with (\ref{eq:global-stabilization-controller}) at $0\in\mathbb{R}^2$ are real since the following inequality holds 
for sufficiently large $k>0$ (i.~e., for $\beta>0$ and sufficiently small $\delta>0$) due to the fact that the left side of \eqref{discriminant} is quadratic in $k$ and the right side linear in $k$. In conclusion, the linearization of the closed-loop system (\ref{eq:ode-system}) with (\ref{eq:global-stabilization-controller}) has a damped response for a small enough $\delta$. That means, in turn, that the oscillations of the predator-prey open-loop motion are completely eliminated, at least locally, for a small enough $\delta$, for any $\varepsilon$ chosen to satisfy the dilution positivity condition \eqref{eq:saturated-control-constraints}.

\section{Regional Stabilization with Positive Dilution} \label{sec:regional-stabilization-positive-dilution}

Now we turn our attention to the study of the stabilizing properties of the feedback law (\ref{eq:global-stabilization-controller}) in the presence of the IDE $\psi$-perturbations \eqref{eq:psi_i}. From \eqref{eq-eta1-perturbed}, \eqref{eq-eta2-perturbed}, \eqref{psihat-defs}, one obtains the perturbed model 
\begin{subequations}
\label{eq-etaeta-perturbed+}
\begin{align} 
    %\dot{\eta}_i = u^* - u + \phi_j(\eta_j+v_j(\psitminusa{j})).
\label{eq-eta1-perturbed+}
   \dot{\eta}_1 &= u^* - u - \phi_2+\tilde\phi_2,\\
\label{eq-eta2-perturbed+}
   \dot{\eta}_2 &= u^* - u + \phi_1 -\tilde\phi_1\,,
\end{align}
\end{subequations}
where we have denoted $\tilde\phi_i = \phi_i - \hat\phi_i$ and have suppressed the arguments $\eta_1,\eta_2 $ for notational brevity. 
With \eqref{eq-eta1-perturbed+}, \eqref{eq-eta2-perturbed+}, and \eqref{eq:global-stabilization-controller}, and \eqref{eq:lyapunov-ode} we get
\begin{eqnarray}
    \dot{V}_1 &=& - \varepsilon(1+\varepsilon)\phi_2^2 -  \frac{\beta \varphi^2}{\sqrt{\delta^2+ \left(\varphi^-\right)^2}} \nonumber \\
    &&+\phi_1\tilde\phi_2 - (1+\varepsilon)\phi_2\tilde\phi_1
    \nonumber\\
    &=&- \varepsilon(1+\varepsilon)\phi_2^2 -  \frac{\beta \varphi^2}{\sqrt{\delta^2+ \left(\varphi^-\right)^2}} \nonumber \\
    &&-(1+\varepsilon)\phi_2 \left(\tilde\phi_1 +\tilde\phi_2\right) +\varphi \tilde\phi_2
    % \nonumber\\
    % &\leq&- \frac{\varepsilon}{2}(1+\varepsilon)\phi_2^2 +\frac{1+\varepsilon}{\varepsilon} \left(\tilde\phi_1^2 +\tilde\phi_2^2\right)
    % -  \frac{\beta \varphi^2}{\sqrt{\delta^2+ \left(\varphi^-\right)^2}}
    %  +\varphi \tilde\phi_2
    \,.
\end{eqnarray} 
From  \eqref{eq:temp-mk} we recall that 
\begin{align}
 \frac{\gamma_1}{\sigma_1}
 D^+ h(G_1)
 +& \frac{\gamma_2}{\sigma_2}
 D^+ h(G_2) \nonumber\\
&\leq - \gamma_1 \left(e^{G_1}-1\right)^2 -  \gamma_2 \left(e^{G_2}-1\right)^2  \,. \label{eq:temp-mk+}
\end{align} 
Now recall the Lyapunov functional \eqref{eq:lyapunov-ode-ide}, for which we obtain
\begin{eqnarray}
\label{Vdot-Iasson}
D^+ V &\leq &    \dot{V}_1 +  \frac{\gamma_1}{\sigma_1}
%e^{\sigma_1 A}}
% \frac{d}{dt}
 D^+ h(G_1)+ \frac{\gamma_2}{\sigma_2}
 %e^{\sigma_2 A}}
% \frac{d}{dt}
 D^+ h(G_2)
\nonumber\\
    &\leq&- \frac{\varepsilon}{2}(1+\varepsilon)\phi_2^2 +\frac{1+\varepsilon}{\varepsilon} \left(\tilde\phi_1^2 +\tilde\phi_2^2\right) \nonumber\\
    &&-  \frac{\beta \varphi^2}{\sqrt{\delta^2+ \left(\varphi^-\right)^2}}
     +\varphi \tilde\phi_2
     \nonumber\\
     && - \gamma_1  \left(e^{G_1}-1\right)^2 -  \gamma_2 \left(e^{G_2}-1\right)^2 
    \,. 
\end{eqnarray} 
Additionally, from \eqref{eq:v_leq_G}, \eqref{phitildes-fulldef}, 
\begin{eqnarray}
\label{phitilde-bound1}
\left| \tilde\phi_1\right| &\leq & \left| \frac{1}{\lambda_1}-\phi_1\right| \ \left| e^{G_1}-1\right|
\\
\label{phitilde-bound2}
\left| \tilde\phi_2\right| &\leq & \left| {\lambda_2}+\phi_2\right| \ \left| e^{G_2}-1\right|\,.
\end{eqnarray}
Let us note next that 
\begin{align}
\label{varphi-inequality}
- \frac{\beta \varphi^2}{\sqrt{\delta^2+ \left(\varphi^-\right)^2}}
     +\varphi \tilde\phi_2 
     \leq& - \frac{\beta}{2} \frac{ \varphi^2}{\sqrt{\delta^2+ \left(\varphi^-\right)^2}} + \frac{\varpi}{2}\tilde\phi_2^2 \nonumber \\
     &+\frac{\varphi^2}{2}\left(-\frac{\beta }{\sqrt{\delta^2+ \left(\varphi^-\right)^2}}  + \varpi\right)
     ,
\end{align}
where we take $\beta >0$ (unlike in Theorem \ref{thm:global-stabilization-controller} where $\beta=0$ is also allowed) and choose $\varpi$ such that
\begin{equation}
0 < \varpi < \beta/\delta\,
\end{equation}
in order to have
\begin{align}
-  \frac{\beta \varphi^2}{\sqrt{\delta^2+ \left(\varphi^-\right)^2}}
     +\varphi \tilde\phi_2 
     \leq& - \frac{\beta}{2} \frac{ \varphi^2}{\sqrt{\delta^2+ \left(\varphi^-\right)^2}} + \frac{\varpi}{2}\tilde\phi_2^2
%     +\frac{\varphi^2}{2}\left(-\frac{\beta }{\sqrt{\delta^2+ \left(\varphi^-\right)^2}}  + \varpi\right) 
\nonumber \\
\qquad \mbox{whenever}  \quad
\varphi(\eta) \geq& - \sqrt{\frac{\beta^2}{\varpi^2}-\delta^2}\,.
\end{align}
From \eqref{Vdot-Iasson}, \eqref{phitilde-bound1}, \eqref{phitilde-bound2}, \eqref{varphi-inequality}, we get
\begin{align}
D^+ V&  \nonumber\\ 
\leq &- \frac{\varepsilon}{2(1+\varepsilon)}\phi_2^2 
+\frac{1+\varepsilon}{\varepsilon}  \left(\frac{1}{\lambda_1} - \phi_1\right)^2
\left(e^{G_1}-1\right)^2\nonumber\\
&+\frac{1+\varepsilon}{\varepsilon} 
        \left(\lambda_2 + \phi_2\right)^2
\left(e^{G_2}-1\right)^2 
\nonumber\\
&
- \frac{\beta}{2} \frac{ \varphi^2}{\sqrt{\delta^2+ \left(\varphi^-\right)^2}} 
     +\frac{\varphi^2}{2}\left(-\frac{\beta }{\sqrt{\delta^2+ \left(\varphi^-\right)^2}}  + \varpi\right)  \nonumber\\  
&+\frac{1}{2\varpi} 
        \left(\lambda_2 + \phi_2\right)^2
\left(e^{G_2}-1\right)^2 
\nonumber\\
& - \gamma_1  \left(e^{G_1}-1\right)^2 -  \gamma_2 \left(e^{G_2}-1\right)^2 
\nonumber\\
=&
- \frac{\varepsilon}{2(1+\varepsilon)}\phi_2^2 
- \frac{\beta}{2} \frac{ \varphi^2}{\sqrt{\delta^2+ \left(\varphi^-\right)^2}} 
\nonumber\\
&
+\left[\frac{1+\varepsilon}{\varepsilon}  \left(\frac{1}{\lambda_1} - \phi_1\right)^2-\gamma_1\right]
\left(e^{G_1}-1\right)^2
\nonumber\\
&
+\left[\frac{1+\varepsilon}{\varepsilon} 
        \left(\lambda_2 + \phi_2\right)^2
        +\frac{1}{2\varpi} - \gamma_2\right]
\left(e^{G_2}-1\right)^2 
\nonumber\\
& 
     +\frac{\varphi^2}{2}\left(-\frac{\beta }{\sqrt{\delta^2+ \left(\varphi^-\right)^2}}  + \varpi\right)    
    \,. 
\end{align}
In the set
\begin{subequations}
\label{eta-sets-allowed}
\begin{align}
    \eta_1 &\geq -\ln\left(\lambda_1\sqrt{ \frac{  \varepsilon}{1+\varepsilon}\frac{\gamma_1}{2}}\right)=
    :- H_1,\\
    \eta_2 &\leq \ln\left(\frac{1}{\lambda_2}\sqrt{\frac{  \varepsilon}{1+\varepsilon}\frac{\gamma_2 - 1/\varpi }{2}}\right)=:H_2.
\end{align}
\end{subequations}
\begin{equation}
\label{varphi-condition}
    \phi_1(\eta_1) + \left(1+\varepsilon\right) \phi_2(\eta_2) \geq - \sqrt{\frac{\beta^2}{\varpi^2}-\delta^2}\,,
\end{equation}
where we note that $\eta_1 = \eta_2 = G_1 = G_2 =0$ is strictly inside, 
the following holds,
\begin{eqnarray}
\label{final-Vdot}
D^+ V &\leq & 
- \frac{\varepsilon}{2(1+\varepsilon)}\phi_2^2 
- \frac{\beta}{2} \frac{ \varphi^2}{\sqrt{\delta^2+ \left(\varphi^-\right)^2}} 
\nonumber\\
&&
-\frac{\gamma_1}{2}
\left(e^{G_1}-1\right)^2
-\frac{\gamma_2}{2}
\left(e^{G_2}-1\right)^2 \,, 
\end{eqnarray}
provided the analysis constants $\varpi,\gamma_1,\gamma_2$ are chosen such that
\begin{eqnarray}
\label{ROA1}
&0 < \varpi < \displaystyle \frac{\beta}{\delta} & \\
\label{ROA2}
&\gamma_1 > \displaystyle\frac{2}{\lambda_1^2}\frac{1+\varepsilon}{\varepsilon} > 0 &\\
\label{ROA3}
&\gamma_2 >  2\lambda_2^2 \displaystyle\frac{1+\varepsilon}{\varepsilon}+\frac{1}{\varpi}
> 2\lambda_2^2 \displaystyle\frac{1+\varepsilon}{\varepsilon}+\frac{\delta}{\beta}
>0 \,.&
\end{eqnarray}
Since \eqref{final-Vdot} is negative definite within the set \eqref{ROA1}, \eqref{ROA2}, \eqref{ROA3}, then the largest level set of $V$ is an estimate of the region of attraction of $\eta=0, \psi\equiv 0$.

The set \eqref{varphi-condition} is difficult to imagine, even with the definitions \eqref{phi-def}. However, in the ``exponentiated'' variables, $q_1 = {\rm e}^{\eta_1}-1, q_2 = {\rm e}^{\eta_2}-1$, this set is written as 
\begin{align}
\label{hyperbola-H}
q_2 &\geq \mathcal{H}(q_1) \ \text{with} \nonumber\\
\mathcal{H}(q_1) &:= \frac{1}{(1+\varepsilon)\lambda_1\lambda_2} 
\left[ \frac{1}{1+q_1}- 
\left(1+ \lambda_1\sqrt{\frac{\beta^2}{\varpi^2}-\delta^2}\right)\right]\,,
\end{align}
where $q_2 = \mathcal{H}(q_1)$ is a hyperbola, and the origin $q_1=q_2 = 0$ (i.e., $\eta_1=\eta_2 = 0$) is within the set \eqref{hyperbola-H}.

\begin{theorem} \label{thm:Iasson}
Let Assumption~\ref{assump:technical_assumption} hold. Consider the closed-loop system \eqref{eq-etaeta-perturbed+}, i.e., system (\ref{eq:transformed-system}) with the control law \eqref{eq:global-stabilization-controller},  
on the state space $\mathbb{R}^2 \times \mathcal{S}$. Let  $\delta>0$ and the parameters $\varepsilon, \beta >0 $ satisfy the conditions \eqref{eq:saturated-control-constraints}. Define $V(\eta,\psi)$ by means of \eqref{eq:lyapunov-ode-ide} for $\gamma_1, \gamma_2$ that satisfy \eqref{ROA2}, \eqref{ROA3} for some $\varpi\in \left(0, \frac{\beta}{\delta}\right)$. Then all the conclusions of Theorem \ref{thm:local-stability}
hold with $\mathcal{D}$ replaced by 
\begin{align}
\label{eq:constraints+}
    \overline{\mathcal{D}} := \Biggl\{ &(\eta,\psi) \in \mathbb{R}^2\times \mathcal{S} \ 
       \Bigg| \nonumber\\
       & \ \eta_1\geq-H_1\,, 
  \ \ 
    \eta_2\leq H_2\,, 
    \nonumber\\
    &  \ \ %u= 
    \phi_1(\eta_1) + (1+\varepsilon)\phi_2(\eta_2) >
 - \sqrt{\frac{\beta^2}{\varpi^2}-\delta^2} \Biggr\}\,,
\end{align}
where $H_1,H_2$ are defined by  \eqref{eta-sets-allowed}.
\end{theorem}

By noting that the constants $H_1, H_2$ in \eqref{eta-sets-allowed} increase with $\gamma_1, \gamma_2$, one is tempted to hope that a compact estimate of the region of attraction of $(\eta_1,\eta_2,G_1,G_2)=0$ may be arbitrarily large. However, by examining the dependence of the Lypunov function $V$ in \eqref{eq:lyapunov-ode-ide} on $\gamma_1, \gamma_2$, as well as the dependence on $(\eta_1,\eta_2) $ of the Lyapunov function $V_1$ defined by \eqref{eq:lyapunov-ode} and \eqref{Phi-defs}, one realizes that it is not only impossible to make an estimate of the region of attraction arbitrarily large by increasing $\gamma_1, \gamma_2$, but that such an increase, while expanding the estimate in $(\eta_1,\eta_2)$, shrinks the estimate in $(G_1,G_2)$. In other words, there is a tradeoff between the allowed size of the initial state $\eta$ and the allowed size of the initial profile of the ``internal age-structured perturbation'' $\psi$. All this is not a consequence of a conservative analysis that we conduct. It is a consequence of the engineering and physical requirement that dilution be positive, which dictates the use of saturated feedback and, ultimately, results in the lack of global robustness to the $\psi$-perturbation.

\section{Simulations}\label{sec:simulations}
In this section, we present a numerical example of the interacting population system and show simulations of both control laws for different initial conditions.

\subsection{Model Parameters and Equilibria}

The age-dependent kernels used for simulation are 
\begin{subequations}\label{eq:parameters}
    \begin{align}
	\mu_i(a)&=\bar{\mu}_i\mathrm{e}^a \\
    k_i(a)&=\bar{k}_i\mathrm{e}^{-a} \\
	g_i(a)&=\bar{g}_i\left(a-a^2\right), 
\end{align}
\end{subequations}
which are biologically inspired and have led to realistic simulation results when studying the behavior of one species in a bioreactor \cite{kurth2023control}. Mortality rates $\mu_i(a)$ increase exponentially with age whereas fertility rates $k_i(a)$ decrease exponentially with age. The maximum interaction is achieved in the middle of the age interval with maximum age $A=1$.
More precisely, we choose both species to exhibit the same behavior, namely $\bar{\mu}_i = 0.5$, $\bar{k}_i = 3$, $\bar{g}_i =0.4$. With these species characteristics, we obtain the parameters $\zeta_1=\zeta_2=1.17$ from Lemma~\ref{lem:ls}.
%, representing the balance between birth, mortality, and interaction. 
A species without external inputs adding population can only survive if $\zeta_i > 0 $.
We choose $u^*=0.15$, which corresponds approximately to $\zeta_1-1$, and obtain $\lambda_1=0.98$, $\lambda_2=1.02$ according to constraint (\ref{eq:u_star_constraint}). From that, the initial values of the steady states follow to $x_1^*(0)=33.81$, $x_2^*(0)=35.19$ by (\ref{eq:ic_constraint}).

We have the freedom to choose $x_{i,0}(a)$, and from that $\eta_i(0)$, $\psi_{i,0}(a)$ follow by transformation (\ref{eq:system-trafo}). 
Two interesting scenarios are 1.) an initially overpopulated prey/underpopulated predator case ($\eta_0$ in fourth quadrant) and 2.) an initially underpopulated prey/overpopulated predator case ($\eta_0$ in second quadrant). Recall that $\eta_{i}>0$ means $x_{i} > x_i^*$, and $\eta_{i}<0$ means $x_{i} < x_i^*$. 
The initial conditions chosen in the following simulations are
    \begin{align} 
    x_{0,\mathrm{FQ}} &= 
    \begin{bmatrix}
        x_1^*(a) e^{1 + 2a} \\
        x_2^*(a) e^{-1 - 2a}
    \end{bmatrix} \label{eq:ic-1} \\
    x_{0,\mathrm{SQ}} &=
\begin{bmatrix}
    x_1^*(a) e^{-1 - 2a} \\
    x_2^*(a) e^{1 + 2a}
\end{bmatrix} \label{eq:ic-2}
    \end{align}
resulting in $\eta_{0,\mathrm{FQ}}=[1.57,-1.41]$ and $\eta_{0,\mathrm{SQ}}=[-1.41, 1.57]$.

\subsection{Results}
We recall from Figure~\ref{fig:simulation-uncontrolled-ic1} that the open-loop system is marginally stable. Simulations of the autonomous $\psi$-dynamics are omitted for the controlled system as they do not change with respect to the uncontrolled system simulations.

We refer to the initial control law (\ref{eq:control_law}) from Sections~\ref{sec:control-design} and \ref{sec:stability-nonzero-psi} which does not ensure positive dilution $u(t)$ as \emph{control A}. The parameters of said control (\ref{eq:control_law}) are chosen to be $\varepsilon=0.2$, $\beta=0.6$. Figures~\ref{fig:control-1-ic1} and Figure~\ref{fig:control-1-ic2} show simulations of system~(\ref{eq:karafyllis_system}) with \emph{control A} (\ref{eq:control_law}) for the ICs (\ref{eq:ic-1}) and (\ref{eq:ic-2}), respectively.

We refer to the enhanced, restrained control law (\ref{eq:global-stabilization-controller}) from Sections~\ref{sec:stability-positive-dilution} and \ref{sec:regional-stabilization-positive-dilution} which ensures positive dilution $u(t)$ as \emph{control B}. The parameters of said control (\ref{eq:global-stabilization-controller}) are chosen to be $\delta=0.2$, $\beta=0.13$, $\varepsilon=0.01$ such that constraint \eqref{eq:saturated-control-constraints} is met. Figures~\ref{fig:control-2-ic1} and \ref{fig:control-2-ic2} show simulations of system (\ref{eq:karafyllis_system}) with \emph{control B} (\ref{eq:global-stabilization-controller}), for the ICs (\ref{eq:ic-1}) and (\ref{eq:ic-2}), respectively.

% CONTROL 1 - IC 1
\begin{figure*}
    \centering
    \begin{minipage}{0.49\textwidth}
    \includegraphics[width=\columnwidth]{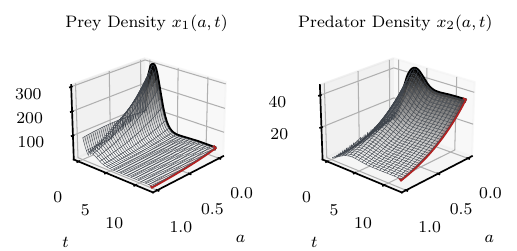}
    \end{minipage}
    \begin{minipage}{0.49\textwidth}
    \includegraphics[width=\columnwidth]{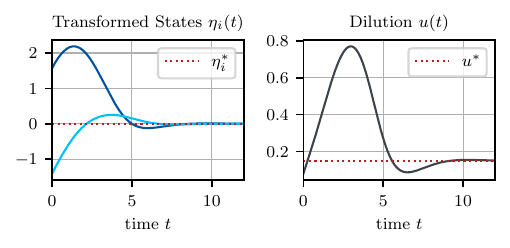}
    \end{minipage}
    \caption{\emph{Control A with initially underpopulated predator and overpopulated prey}: Population densities $x_i$ from system~(\ref{eq:karafyllis_system}) with steady-state $x_i^*(a)$ (\mlLineLegend{rot}), and transformed state variables $\eta_1$ (\mlLineLegend{dunkelblau}) and $\eta_2$ (\mlLineLegend{hellblau}) from representation (\ref{eq:transformed-system}) under control law (\ref{eq:control_law}), and with parameter set (\ref{eq:parameters}), ICs (\ref{eq:ic-1}), $\varepsilon=0.2$, $\beta=0.6$.}
    \label{fig:control-1-ic1}
\end{figure*}

% CONTROL 2 - IC 1
\begin{figure*}
    \centering
    \begin{minipage}{0.49\textwidth}
    \includegraphics[width=\columnwidth]{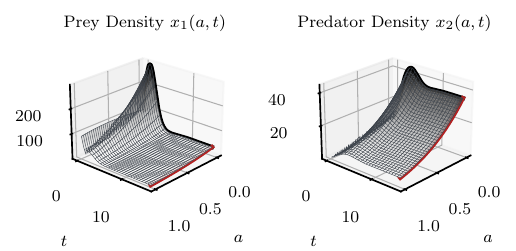}
    \end{minipage}
    \begin{minipage}{0.49\textwidth}
    \includegraphics[width=\columnwidth]{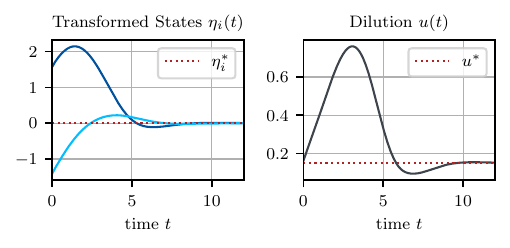}
    \end{minipage}
    \caption{\emph{Control B with initially underpopulated predator and overpopulated prey}: Population densities $x_i$ from system~(\ref{eq:karafyllis_system}) with steady-state $x_i^*(a)$ (\mlLineLegend{rot}), and transformed state variables $\eta_1$ (\mlLineLegend{dunkelblau}) and $\eta_2$ (\mlLineLegend{hellblau}) from representation (\ref{eq:transformed-system})  under control law (\ref{eq:global-stabilization-controller}), and with parameter set (\ref{eq:parameters}), ICs (\ref{eq:ic-1}), $\delta=0.2$, $\beta=0.13$, $\varepsilon=0.01$.}
    \label{fig:control-2-ic1}
\end{figure*}

% CONTROL 1 - IC 2
\begin{figure*}
    \centering
    \begin{minipage}{0.49\textwidth}
    \includegraphics[width=\columnwidth]{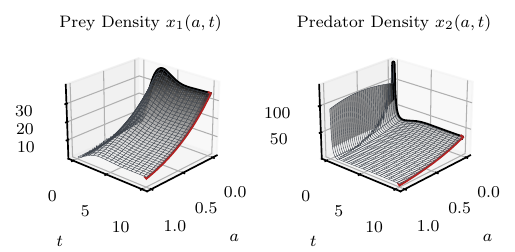}
    \end{minipage}
    \begin{minipage}{0.49\textwidth}
    \includegraphics[width=\columnwidth]{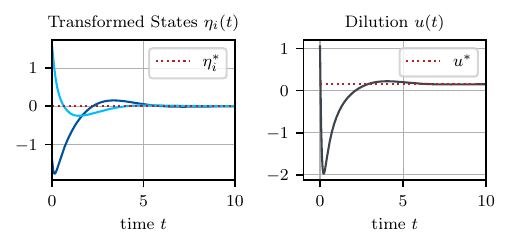}
    \end{minipage}
    \caption{\emph{Control A with initially underpopulated prey and overpopulated predator}: Population densities $x_i$ from system (\ref{eq:karafyllis_system}) with steady-state $x_i^*(a)$ (\mlLineLegend{rot}), and transformed state variables $\eta_1$ (\mlLineLegend{dunkelblau}) and $\eta_2$ (\mlLineLegend{hellblau}) from representation (\ref{eq:transformed-system}) under control law (\ref{eq:control_law}), and with parameter set (\ref{eq:parameters}), ICs (\ref{eq:ic-2}), $\varepsilon=0.2$, $\beta=0.6$.}
    \label{fig:control-1-ic2}
\end{figure*}

% CONTROL 2 - IC 2
\begin{figure*}
    \centering
    \begin{minipage}{0.49\textwidth}
    \includegraphics[width=\columnwidth]{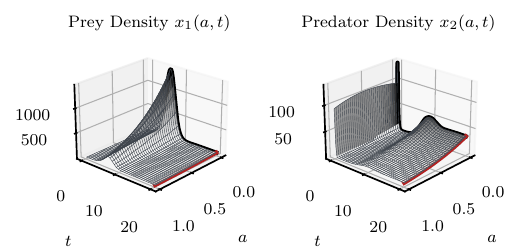}
    \end{minipage}
    \begin{minipage}{0.49\textwidth}
    \includegraphics[width=\columnwidth]{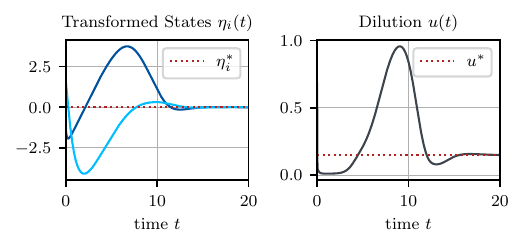}
    \end{minipage}
    \caption{\emph{Control B with initially underpopulated prey and overpopulated predator}: Population densities $x_i$ from system (\ref{eq:karafyllis_system}) with steady-state $x_i^*(a)$ (\mlLineLegend{rot}), and transformed state variables $\eta_1$ (\mlLineLegend{dunkelblau}) and $\eta_2$ (\mlLineLegend{hellblau}) from representation (\ref{eq:transformed-system})  under control law (\ref{eq:global-stabilization-controller}), and with parameter set (\ref{eq:parameters}), ICs (\ref{eq:ic-2}), $\delta=0.2$, $\beta=0.13$, $\varepsilon=0.01$.}
    \label{fig:control-2-ic2}
\end{figure*}

\subsection{Discussion}
Both proposed controllers achieve convergence to the desired steady-states $x_i^*$. 
While the dilution input $u$ stays positive at all times with both proposed control laws for IC \eqref{eq:ic-1} where the predator is initially underpopulated (cf. Figures \ref{fig:control-1-ic1} and \ref{fig:control-2-ic1}), differences can be observed when the predator is initially advantaged (cf. Figures \ref{fig:control-1-ic2} and \ref{fig:control-2-ic2}).   

\subsubsection{Initially underpopulated predator}
\emph{Control A} and \emph{control B} take nearly identical values, which are positive at all times -- harvesting both species is favorable to diminish prey density and increase predator density: the dilution and predator work in tandem relative to the prey population. The steady state is reached after around \qty{8}{\hour}.

\subsubsection{Initially overpopulated predator}
In this case, the enhancement of \emph{control B} with respect to \emph{control A} is pointed out: \emph{Control A} takes negative values when the predator is initially overpopulated (Figure~\ref{fig:control-1-ic2}), whereas \emph{control B} remains positive at all times (Figure~\ref{fig:control-2-ic2}).
\emph{Control A} taking negative values corresponds to adding population, to compensate for the fact that, when the prey population is depleted, positive dilution may result in overharvesting the prey and, consequently, in the extinction of both populations. In this case, \emph{control B} waits until the prey density increases by its natural course before taking control action.
The steady state is reached after \qty{15}{\hour} (whereas employing the unrealistic positive dilution would drive the system to its steady state more than twice as fast).

\section{Conclusion}\label{sec:conclusion}
The interest in extending the foundational design and analysis results in \cite{karafyllis2017stability} goes in many directions. In this paper we extended the single population model to two interacting populations in a predator-prey setup. 
We developed two control designs, a modified Volterra-like CLF which employs possibly negative harvesting, and a more sophisticated, restrained controller with positive harvesting. Both controllers stabilize the ODE system globally asymptotically and locally exponentially. For the ODE-IDE system, the globality is lost with both controllers, but regional stability holds.
Such generalizations of the chemostat problem into multi-population systems, of which epidemiology is one possible application, open exciting possibilities for future research directions, such as investigating predator-predator scenarios or more than two interacting populations.

\section*{References}
\bibliography{main.bib}
\bibliographystyle{IEEEtranS}

\begin{IEEEbiography}[{\includegraphics[width=1in,height=1.25in,clip,keepaspectratio]{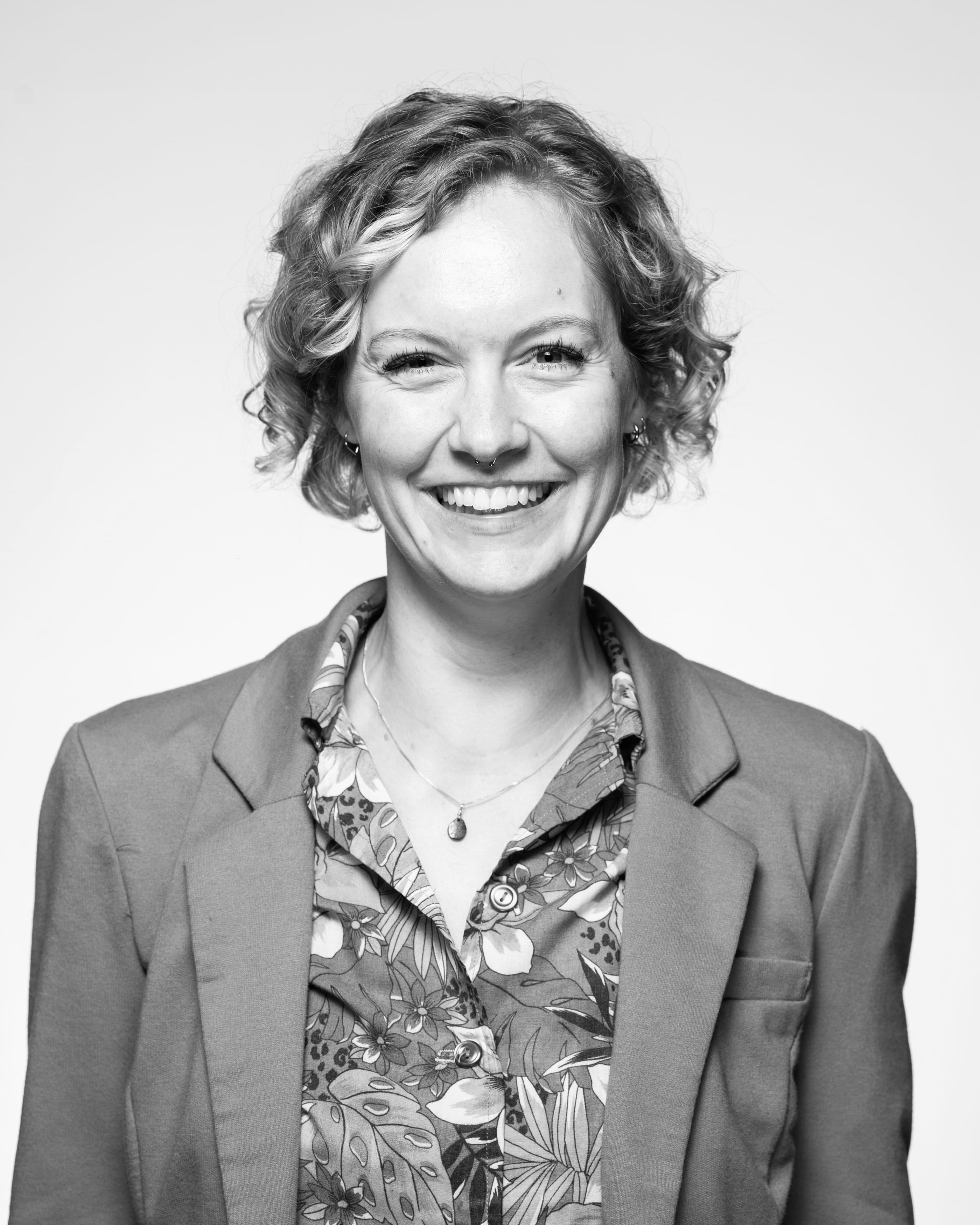}}]{Carina Veil} (Member, IEEE) is a postdoctoral researcher with Stanford University in the field of complex systems and control. She received a B.Sc in biomedical engineering, M.Sc. in engineering cybernetics, and Ph.D in mechanical engineering from the University of Stuttgart, Germany, in 2017, 2020, and 2023, respectively. For her Ph.D., she investigated impedance-based tissue differentiation for tumor detection. From 2023 to 2025, she has been a postdoctoral researcher at the Institute for System Dynamics, University of Stuttgart, Germany working on PDE control in the context of population systems. In 2024, she has been a Visiting Scholar at the University of California San Diego, USA, with Prof. Miroslav Krsti\'c. Her research interests include applications and methods at the intersection of control, nature, and biomedical engineering.
\end{IEEEbiography}

\begin{IEEEbiography}[{\includegraphics[width=1in,height=1.25in,clip,keepaspectratio]{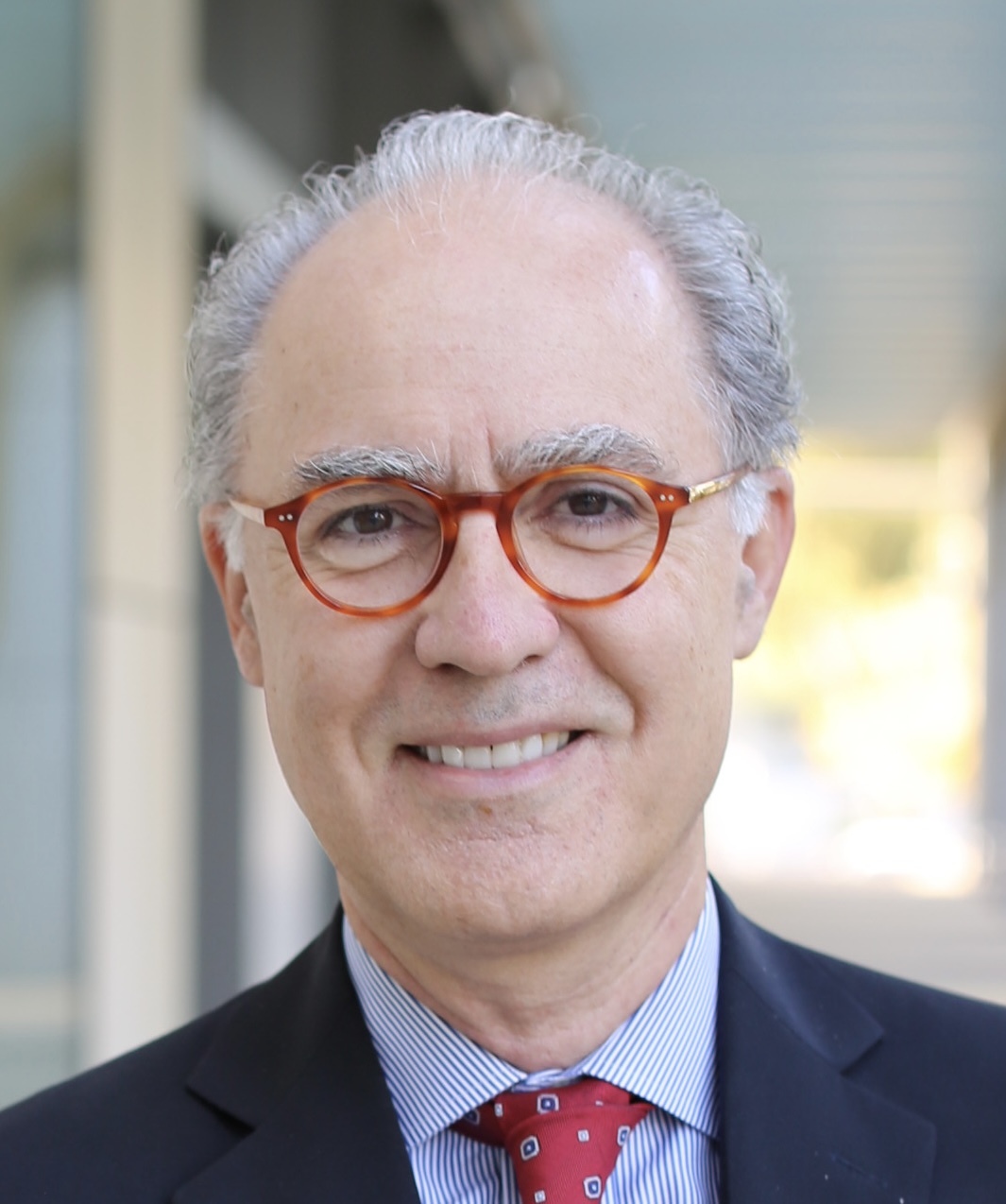}}]{Miroslav Krsti\'c} (Fellow, IEEE) is Distinguished Professor of Mechanical and Aerospace Engineering, holds the Alspach endowed chair, and is the founding director of the Center for Control Systems and Dynamics at UC San Diego. He also serves as Senior Associate Vice Chancellor for Research at UCSD. As a graduate student, Krsti\'c won the UC Santa Barbara best dissertation award and student best paper awards at CDC and ACC. Krsti\'c has been elected Fellow of seven scientific societies - IEEE, IFAC, ASME, SIAM, AAAS, IET (UK), and AIAA (Assoc. Fellow) - and as a foreign member of the Serbian Academy of Sciences and Arts and of the Academy of Engineering of Serbia. He has received the Richard E. Bellman Control Heritage Award, Bode Lecture Prize, SIAM Reid Prize, ASME Oldenburger Medal, Nyquist Lecture Prize, Paynter Outstanding Investigator Award, Ragazzini Education Award, IFAC Nonlinear Control Systems Award, IFAC Ruth Curtain Distributed Parameter Systems Award, IFAC Adaptive and Learning Systems Award, Chestnut textbook prize, AV Balakrishnan Award for the Mathematics of Systems, Control Systems Society Distinguished Member Award, the PECASE, NSF Career, and ONR Young Investigator awards, the Schuck (’96 and ’19) and Axelby paper prizes, and the first UCSD Research Award given to an engineer. Krsti\'c has also been awarded the Springer Visiting Professorship at UC Berkeley, the Distinguished Visiting Fellowship of the Royal Academy of Engineering, the Invitation Fellowship of the Japan Society for the Promotion of Science, and four honorary professorships outside of the United States. He serves as Editor-in-Chief of Systems \& Control Letters and has been serving as Senior Editor in Automatica and IEEE Transactions on Automatic Control, as editor of two Springer book series, and has served as Vice President for Technical Activities of the IEEE Control Systems Society and as chair of the IEEE CSS Fellow Committee. Krsti\'c has coauthored eighteen books on adaptive, nonlinear, and stochastic control, extremum seeking, control of PDE systems including turbulent flows, and control of delay systems.
\end{IEEEbiography}

\begin{IEEEbiography}[{\includegraphics[width=1in,height=1.25in,clip,keepaspectratio]{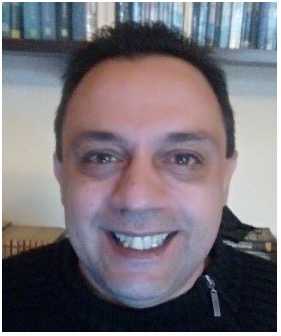}}]{Iasson Karafyllis} is a Professor in the Department of Mathematics, NTUA, Greece. He is a coauthor (with Z.-P. Jiang) of the book Stability and Stabilization of Nonlinear Systems, Springer-Verlag London, 2011 and a coauthor (with M. Krstic) of the books Predictor Feedback for Delay Systems: Implementations and Approximations, Birkhäuser, Boston 2017 and Input-to-State Stability for PDEs, Springer-Verlag London, 2019. Since 2013 he is an Associate Edi-tor for the International Journal of Control and for the IMA Journal of Mathematical Control and Information. Since 2019 he is an Associate Editor for Systems and Control Letters and Mathematics of Control, Signals and Systems. His research interests include mathematical control theory and nonlinear systems theory.
\end{IEEEbiography}

\begin{IEEEbiography}[{\includegraphics[width=.85in,height=1.2in,clip,keepaspectratio]{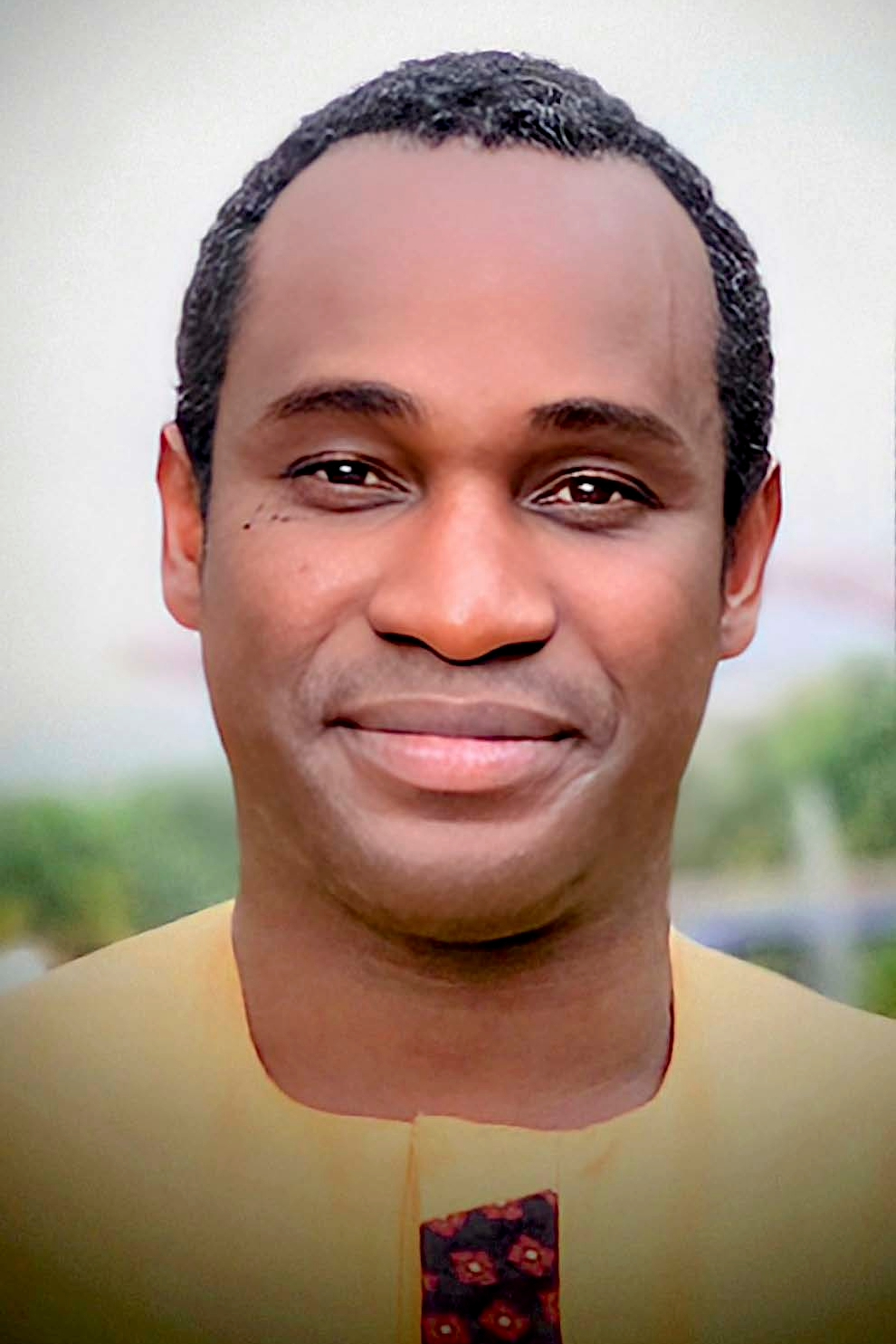}}]{Mamadou Diagne} is an Associate Professor with the Department of Mechanical and Aerospace Engineering at the University of California San Diego, affiliated with the Department of Electrical and Computer Engineering. He completed his Ph.D. in 2013 at the Laboratory of Control and Chemical Engineering of the University Claude Bernard Lyon I in Villeurbanne (France). Between 2017 and 2022, he was an Assistant Professor with the Department of Mechanical Aerospace and Nuclear Engineering at Rensselaer Polytechnic Institute, Troy, New York. From 2013 to 2016, he was a Postdoctoral Scholar, first at the University of California San Diego, and then at the University of Michigan. His research focuses on the control of distributed parameters systems (DPS) and delay systems as well as coupled PDEs and nonlinear ordinary differential equations with a particular emphasis on adaptive control, sampled-data control and event-triggered control. He serves as Associate Editor for Automatica and Systems $\&$  Control Letters. He is Vice-Chair for Industry of the IFAC Technical Committee on Adaptive and Learning Systems. He is a member of the Task Force on Diversity, Outreach $\&$ Development Activities (DODA) of the IEEE Control System Society. He received the NSF Career Award in 2020.
\end{IEEEbiography}

\begin{IEEEbiography}[{\includegraphics[width=1in,height=1.25in,clip,keepaspectratio]{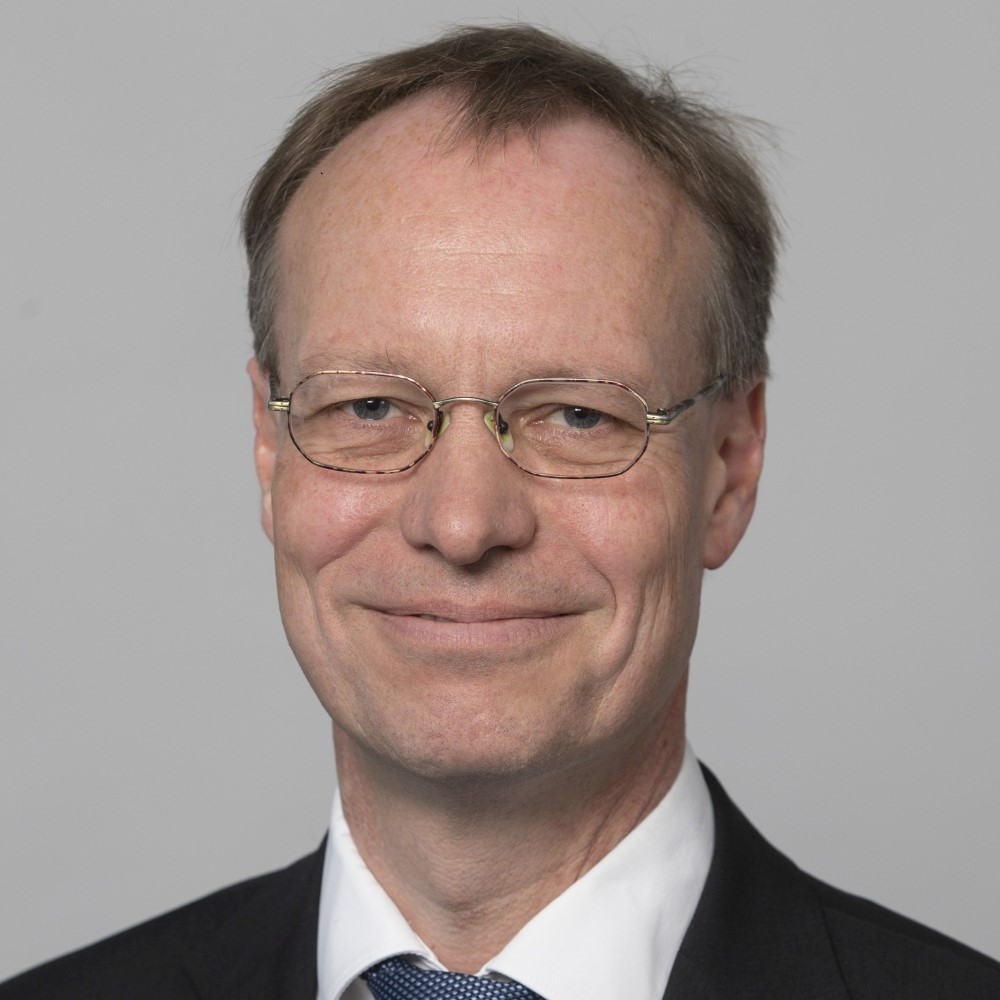}}]{Oliver Sawodny} (Senior Member, IEEE) received the Dipl.-Ing. degree in electrical engineering from the University of Karslruhe, Karlsruhe, Germany, in 1991, and the Ph.D. degree from Ulm University, Ulm, Germany, in 1996. In 2002, he became a Full Professor with the Technical University of Illmenau, Illmenau, Germany. Since 2005, he has been the Director of the Institute for System Dynamics, University of Stuttgart, Stuttgart, Germany. His current research interests include methods of differential geometry, trajectory generation, and applications to mechatronic systems.
\end{IEEEbiography}

\end{document}